\documentclass[aps,prr,twocolumn,floats,final,superscriptaddress,citeautoscript,longbibliography]{revtex4-2}
\usepackage{array}
\usepackage{amsmath,amssymb, amsthm, amsmath, bm, braket}
\usepackage{graphicx}
\usepackage{multirow}
\usepackage[colorlinks=true,linkcolor=blue,citecolor=blue,urlcolor=blue]{hyperref}
\usepackage[svgnames]{xcolor}
\usepackage{tikz}
\usepackage{titlesec}
\newtheorem{theorem}{Theorem}
\newcommand{\ketbra}[2]{\mbox{$| #1 \rangle \langle #2 |$}}     
\newcommand{\proj}[1]{\mbox{$|#1\rangle \langle #1 |$}}       

\begin{document}

\title{Bounding Entanglement Entropy Using Zeros of Local Correlation Matrices}
\author{Zhiyuan Yao}
\affiliation{Institute for Advanced Study, Tsinghua University, Beijing 100084, China}
\author{Lei Pan}
\affiliation{Institute for Advanced Study, Tsinghua University, Beijing 100084, China}
\author{Shang Liu}
\email{sliu.phys@gmail.com}
\affiliation{Kavli Institute for Theoretical Physics, University of California, Santa Barbara, California 93106, USA}
\author{Pengfei Zhang}
\email{PengfeiZhang.physics@gmail.com}
\affiliation{Institute for Quantum Information and Matter and Walter Burke Institute for Theoretical Physics, California Institute of Technology, Pasadena, California 91125, USA}

\date{\today}

\begin{abstract}

Correlation functions and entanglement are two different aspects to characterize quantum many-body states. While many correlation functions are experimentally accessible, entanglement entropy (EE), the simplest characterization of quantum entanglement, is usually difficult to measure. In this Letter, we propose a protocol to bound EE by local measurements. This protocol utilizes local correlation matrices and focuses on their (approximate) zero eigenvalues. Given a quantum state, each (approximate) zero eigenvalue can be used to define a set of local projection operators. An auxiliary Hamiltonian can then be constructed by summing these projectors.
When the construction only involves projectors of zero eigenvalues, we prove the EE of a subsystem is bounded by the ground-state degeneracy of the auxiliary Hamiltonian on this subsystem.
When projectors from nonzero eigenvalues are included, we show the EE can be bounded by a thermal entropy of the subsystem. Our protocol can be applied experimentally to investigate exotic quantum many-body states prepared in quantum simulators.

\end{abstract}

\maketitle

\textit{\color{blue}Introduction.--} Entanglement not only plays an essential role in quantum information science but also becomes more and more important in modern studies of quantum matters \cite{Amico:2008RMP,Laflorencie:2016EE,Eisert:2010ca}. Entanglement entropy (EE) can characterize quantum thermalization \cite{Abanin:2019cm}, describe quantum critical \cite{Vidal:2003ei,Sachdev:2008qm,HOLZHEY1994443,Calabrese2004,Calabrese2009} and topologically ordered states \cite{Kitaev:2006PRL_EE,Levin:2006PRL_EE}, and bridge quantum correlation and geometric metric in holographic quantum matter \cite{Nishioka:2009he,Zaanen:Book,Hartnoll:Book}. On the other hand, conventional characterizations of quantum matters are based on correlation functions and physical observables \cite{Chaikin:2000Book}. It is, therefore, natural to seek connections between EE and correlation functions or physical observables.

There are several well-established connections. It is known that the mutual information can be used to bound correlation functions \cite{Wolf:2008PRL,Xiaoliang_talk}. The scrambling of entanglement is also directly related to the out-of-time-ordered correlator \cite{Fan:2017OTOC,Swingle:2016PRA,Chen:2016ul,Chen:2017OTOC,Huang:2017OTOC,He:2017PRB}. Experimentally, the R\'enyi entropy can be determined by measuring the shift operator between multiple copies of the target quantum state \cite{Ekert:2002de,Alves:2004me,Daley:2012me,Islam:2015cm,Kaufman:2016qt} or averaging over randomized measurements on single copies of a given quantum state  \cite{Enk:2012hr,Elben:2018re,Vermersch:2018un,Elben:2019sc,Brydges:2020pr,Satzinger:2021Science}.  Nevertheless, these protocols usually suffer from large statistical fluctuations as the subsystem size increases. Hence, experimental measurements of EEs so far are limited to small subsystem sizes.

In this Letter, we bring out a connection between the bound of the EE and the (approximate) zeros in the correlation matrix. The correlation matrix is first introduced \cite{Qi:2019CorrMat} in the Hamiltonian reconstruction problem of ``recovering" the Hamiltonian from one of its excited states \cite{Garrison:2018da}. Here we introduce a slightly generalized \textit{local} version of the correlation matrix and use it to bound the EE of a state. The local correlation matrix is Hermitian and positive semidefinite with non-negative eigenvalues. We focus on (approximate) zeros among all eigenvalues of the correlation matrix. Each (approximate) zero gives rise to a set of projection operators $\hat{P}_i$ which (approximately) satisfies $\hat{P}_i \ket{\psi} = 0$ and constrains $\ket{\psi}$.
To quantify the extent of these constraints and give an upper bound of the EE, we introduce a positive semidefinite auxiliary Hamiltonian as the sum of these projectors multiplied by positive prefactors. When only projectors of zero eigenvalues are included, we prove that the von Neumann EE is upper bounded by the logarithm of the ground-state degeneracy of the auxiliary Hamiltonian constructed on the subsystem.
When projectors of approximate zeros are included, the EE of a subsystem can be bounded by the thermal entropy at a temperature determined by the subsystem energy of the state.

\begin{figure}[t]
\centering
\includegraphics[width=\linewidth]{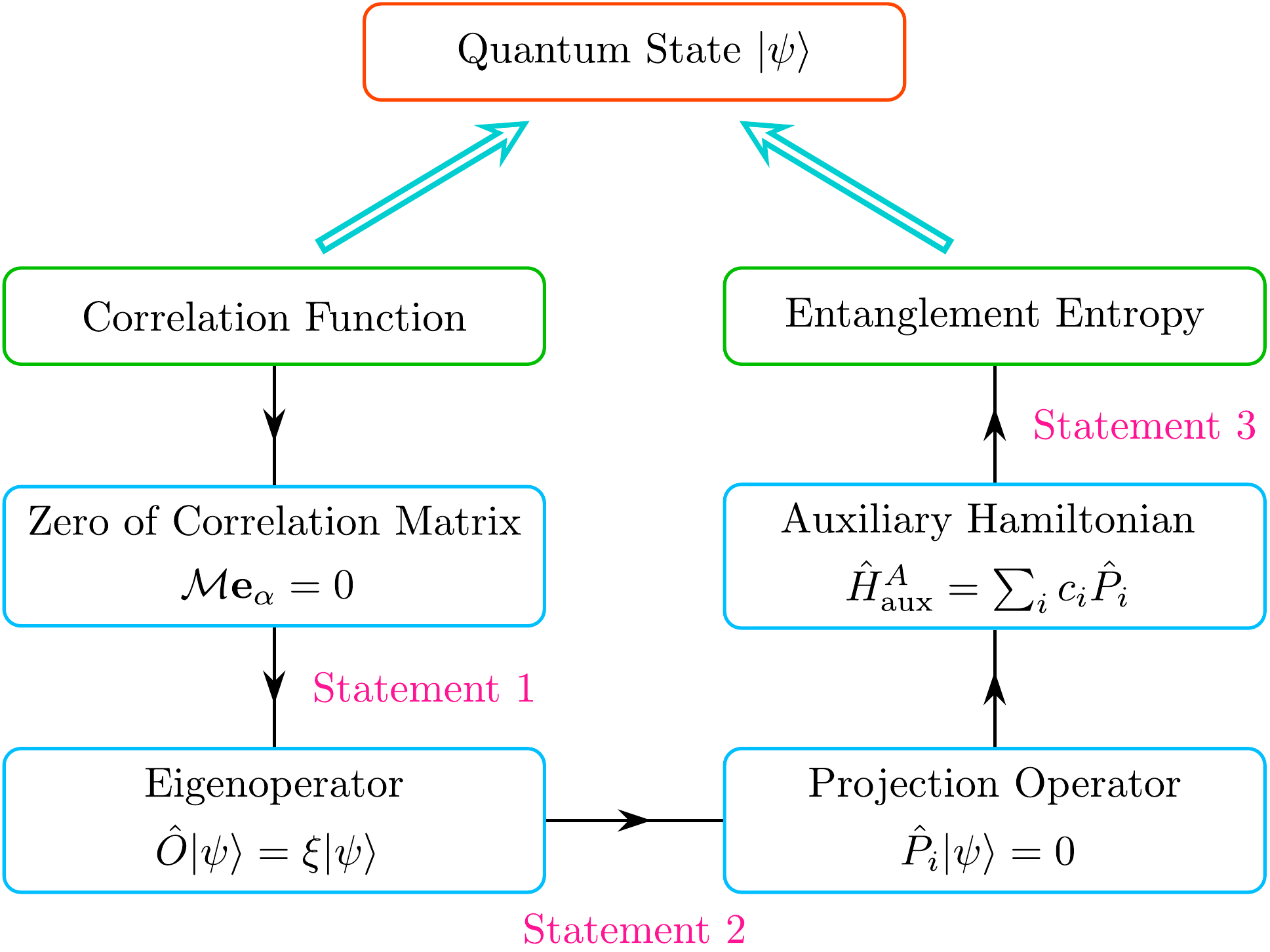}
\caption{Correlation functions and entanglement entropy are two different tools to characterize a quantum many-body state. This work establishes a connection between correlation and entanglement, based on the three statements discussed in the text.}
\label{Schematic}
\end{figure}

\textit{\color{blue}Protocol.--} We first present our protocol, schematically shown in Fig. \ref{Schematic}, of including only projectors from zero eigenvalues in the auxiliary Hamiltonian. It follows step-wisely from the following three statements and connects correlation functions and EE. It is schematically shown in Fig. \ref{Schematic} and is discussed in detail below.

\textbf{Statement 1: from Correlation Matrix to Eigenoperators.} We consider a continuous region of
$k$ sites and a set of linearly-independent $k$-local operators $\{\hat{L}_i\}$ acting only in this region. For a given quantum many-body state $|\psi\rangle$, we introduce its hermitian \textit{correlation matrix} $\mathcal{M}$ with element $\mathcal{M}_{ij}$ defined as
\begin{equation}
\mathcal{M}_{ij}=\langle \psi|\hat{L}^\dag_i\hat{L}_j|\psi\rangle-\langle \psi|\hat{L}^\dag_i|\psi\rangle\langle \psi|\hat{L}_j|\psi\rangle.
\end{equation}
Denote the eigenvalues of $\mathcal{M}$ as $\lambda_\alpha$ and the corresponding eigenvectors as $\bf{e}_\alpha$. For each zero eigenvalue $\lambda_{\alpha}=0$, if exists, one can use $\mathbf{e}_{\alpha}$ to construct an \textit{eigenoperator} $\hat{O}$ of $\ket{\psi}$ where $\hat{O} \ket{\psi} = \xi \ket{\psi}$ for some constant $\xi$.

\begin{proof}
Let us denote ${\bf e}_\alpha=(w_1,w_2,\dots, w_n)$ and construct an operator $\hat{O}=\sum\limits_{i}w_i\hat{L}_i$. Then, if $\lambda_\alpha = 0$ we have
\begin{equation}
\langle\psi|\hat{O}^\dag\hat{O}|\psi\rangle-\langle\psi|\hat{O}^\dag|\psi\rangle\langle\psi|\hat{O}|\psi\rangle=\sum\limits_{ij}w_i^*\mathcal{M}_{ij}w_j=0. \label{psiOO}
\end{equation}
On the other hand, let $\hat{O}|\psi\rangle=\xi|\psi\rangle+|\psi_\perp\rangle$ with $\langle\psi_\perp|\psi\rangle=0$, the above relation translates to
\begin{equation}
0 = \langle\psi|\hat{O}^\dag\hat{O}|\psi\rangle-\langle\psi|\hat{O}^\dag|\psi\rangle\langle\psi|\hat{O}|\psi\rangle=\langle\psi_\perp|\psi_\perp\rangle \, ,
\end{equation}
meaning $\hat{O}|\psi\rangle=\xi|\psi\rangle$. Clearly, the number of independent eigenoperators $\hat{O}$ one can construct in this way equals the geometric/algebraic multiplicity of eigenvalue zero.
\end{proof}

\textbf{Statement 2: from Eigenoperators to Auxiliary Hamiltonian.} From each eigenoperator $\hat{O}$, one can construct a set of projection operators $\{\hat{P}_i\}$ each of which annihilates the state, $\hat{P}_i|\psi\rangle=0$.

\begin{proof}
Let us denote $\hat{\Delta}$ the positive semidefinite operator defined as ($\hat{I}$ being the identity operator)
\begin{equation}\label{eq:Delta}
	\hat{\Delta}=(\hat{O}^\dag-\xi^{*}\hat{I})(\hat{O}-\xi\hat{I}).
\end{equation}
and $\ket{i}$ the basis that diagonalizes $\hat{\Delta}$,
\begin{equation} \label{eq:DeltaP}
	\hat{\Delta}=\sum\limits_{i}\alpha_i|i\rangle\langle i| \equiv \sum\limits_{i}\alpha_i \hat{P}_i,
\end{equation}
where $\alpha_i > 0$. Since $\hat{\Delta}$ annihilates the state and $0 = \bra{\psi} \hat{\Delta} \ket{\psi} = \sum_{i} \alpha_i \bra{\psi} \hat{P}_{i} \ket{\psi} = \sum_{i} \alpha_i \bra{\psi} \hat{P}_{i}^{2} \ket{\psi} = \sum_{i} \alpha_i ||\hat{P}_i\ket{\psi}||^2$, it follows that each $\hat{P}_{i}$ annihilates the state, $\hat{P}_i|\psi\rangle=0$. Moreover, the projectors $\hat{P}_{i}$ are $k$-local operators since $\hat{O}$ is $k$-local.
\end{proof}

Intuitively, the existence of a set of such projectors $\hat{P}_{i}$ means the wave function $\ket{\psi}$ tends to have a simple form and be low entangled. This is in contrast to a generic volume-law eigenstate where such local projectors $\hat{P}_{i}$ do not exist as implied by the uniqueness of the Hamiltonian reconstruction process \cite{Qi:2019CorrMat}. This also means we should find as many such projectors as we can if we want to give a good upper bound on the EE of $\ket{\psi}$ of a subsystem $A$.
To this end, we first exhaust all continuous regions of $k$ sites in $A$ and collect all $k$-local projection operators $\hat{P}_i$ constructed above \footnote{Refer to Appendix A for a visual demonstration.}. We then construct a positive semidefinite auxiliary Hamiltonian on $A$ by summing these projectors multiplied by positive numbers,
\begin{equation}\label{eq:Haux}
    \hat{H}^A_\text{aux} = \sum\limits_{\alpha} \Tilde{c}_\alpha \hat{\Delta}_\alpha = \sum\limits_{i}c_i\hat{P}_i, \quad \Tilde{c}_\alpha > 0, \, c_{i} > 0 \, ,
\end{equation}
where $\alpha$ enumerates $\hat{\Delta}_{\alpha}$ operators of zero eigenvalues of local correlation matrices on all $k$-site regions and $c_i = \sum_{\alpha} \tilde{c}_{\alpha} \alpha_i$ according to Eq.~\eqref{eq:DeltaP}. The introduction of $\hat{H}^A_\text{aux}$ will prove to be convenient for us to discuss the bound of EE.

\textbf{Statement 3: from Auxiliary Hamiltonian to the EE bound---the Ground-State Degeneracy and the Entanglement Entropy (GSD-EE) Theorem.} If $D^A$ denotes the ground-state degeneracy of $\hat{H}^A_\text{aux}$ under the open boundary condition (OBC), the von Neumann EE of $\ket{\psi}$ between this subsystem $A$ and the rest of the system $B$ is upper bounded by $\log D^A$.

\begin{proof}
Because $\hat{P}_i|\psi\rangle=0$ and $\hat{P}_i$ only acts on the subsystem $A$, the reduced density matrix $\hat{\rho}_A=\text{Tr}_\text{B} |\psi\rangle\langle\psi|$ satisfies
\begin{equation}
\text{Tr}_\text{A} (\hat{P}_i\hat{\rho}_A) = \text{Tr}_\text{A}(\text{Tr}_\text{B}\hat{P}_i|\psi\rangle\langle\psi|)=0. \label{TrAPrho}
\end{equation}
Diagonalizing $\hat{\rho}_{A}$ with a set of orthonormal basis $\{\ket{\psi_{k}}\}$, $\hat{\rho}_A=\sum\limits_{k}p_k|\psi_{k}\rangle\langle \psi_{k}|$ where $p_{k} > 0$, we can cast the above relation into the following form
\begin{equation}
\text{Tr}_\text{A} (\hat{P}_i\hat{\rho}_A) = \sum\limits_{k}p_k ||\hat{P}_i|\psi_k\rangle||^2 \, .
\end{equation}
It becomes clear that $\hat{P}_i |\psi_k\rangle=0$ for every $\ket{\psi_{k}}$ and $\hat{P}_{i}$. Therefore, $\ket{\psi_{k}}$ is the ground state of $\hat{H}^{A}_{\text{aux}}$ with zero energy since $\hat{H}^{A}_{\text{aux}}$ is positive semidefinite. Consequently, the number of non-zero $p_k$ is no bigger than the ground-state degeneracy $D^A$ of the $\hat{H}^A_\text{aux}$, and the von Neumann EE is upper bounded by $\log D^A$.
\end{proof}

Bounding EE of a large subsystem $A$ requires the determination of $D^{A}$ on the same large subsystem $A$, and to this end we have specially designed an algorithm  \footnote{Refer to Appendix D for details.}. For a one-dimensional (1D) system, as a corollary of our GSD-EE theorem, if $D^A$ is upper bounded by a power-law of the subsystem size $\sim L_\text{A}^\alpha$, the EE is then upper bounded by $\sim \log L_\text{A}$, indicating the state is sub-volume entangled.
A few remarks are in order before we move on to specific examples. Firstly, the operator space that the set of eigenoperators $\{ \hat{O}_{\alpha} \}$ spans does not depend on the form of $\hat{L}_{i}$ as long as they span the same operator space $\mathcal{V}$. Secondly, the EE upper bound does not depend on the form of $\hat{O}_{\alpha}$, i.e., the choice of $\{\mathbf{e}_{\alpha}\}$, in the case of the zero eigenvalue of $\mathcal{M}$ being degenerate \footnote{see Appendix B for a proof}. Finally, larger $\mathcal{V}$ in general leads to a better upper bound with however, more numerical/experimental effort to obtain $\mathcal{M}$. Consequently, we should make a balance and choose a moderate number of operators. Nevertheless, using 1-local operator basis, our protocol already gives the perfect EE upper bound of zero for product states \footnote{This is proved in Appendix C.}.

\textit{\color{blue}Examples.--}
Below we shall apply the GSD-EE theorem to the scar states in the generalized Affleck--Kennedy--Lieb--Tasaki (AKLT) model and the ground state of the toric code model.

\textbf{Example 1: Scar states in the AKLT model.}
The 1D AKLT model and its generalizations are known to host a series of quantum many-body scar states described the spectrum generating algebra \cite{Moudgalya:2018AKLT,Mark:2020AKLT}. These scar states
admit the following analytical expression,
\begin{equation} \label{eq:AKLT_scars}
    \ket{\psi_{n}} = (\hat{Q}^{\dagger})^{n} \ket{G},
\end{equation}
where $\hat{Q}^{\dagger} = \sum_{i} (-1)^{i} (\hat{S}_{i}^{+})^{2}$ with $i$ being the site index and $\hat{S}_{i}^{+}$ being the spin raising operator at site $i$, and $\ket{G}$ is the ground state of the AKLT model under the periodic boundary condition (PBC).
Here the system size $L$ is assumed to be even, $n =0, 1, \cdots, L/2 -1$ for odd $L/2$, and $n =0, 1, \cdots, L/2$ for even $L/2$.

To construct the correlation matrices, we choose 3-local operators $\{\hat{L}_{i}\}$ that are constructed as tensor products of the Gell-Mann matrices,
\begin{equation}
    \hat{L}_{i} = \lambda_{a}^{(1)} \otimes \lambda_{b}^{(2)} \otimes \lambda_{c}^{(3)} \, ,
\end{equation}
where $a, b, c = 1,\cdots, 8$ are the indices of the eight standard Gell-Mann matrices and the superscripts $1, 2, 3$ denote three successive sites.
As an illustrative example, we choose the state $\ket{\psi_{1}} = \hat{Q}^{\dagger} \ket{G}$ and construct the auxiliary Hamiltonian, following the general protocol described above.
Since our target state is translational invariant up to a sign, we just need to construct the local correlation matrix of arbitrary three neighboring sites: all the other projectors entering the auxiliary Hamiltonian are obtained by simply translating the local projectors of this correlation matrix. It turns out that projectors constructed from $\ket{\psi_{1}}$ also annihilate all the scar states $\ket{\psi_{n}}$ for all system sizes $L>3$ \footnote{Refer to Appendix E for details.}. Thus the EE scaling behavior obtained for our target state $\ket{\psi_{1}}$ will also bound all other scar states $\ket{\psi_{n}}$.
Using our specially designed algorithm, we are able to determine the ground-state degeneracy $D^A$ of the auxiliary Hamiltonian up to system size $L_{\text{A}}$ of a few hundred.
We numerically find that $D^A=2L_\text{A}+2$, which means the EE at most scales with the logarithm of the subsystem size.

\begin{figure}[t]
\centering
\includegraphics[width=0.85\linewidth]{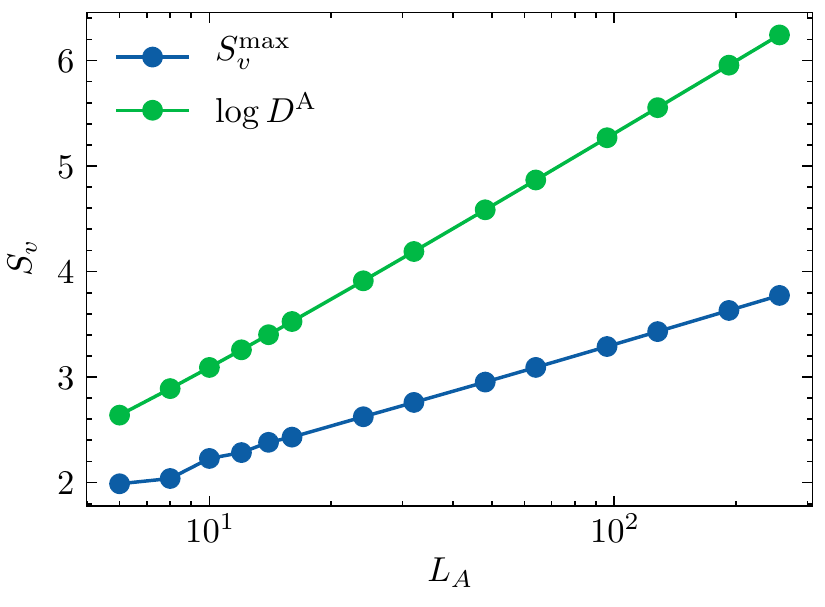}
\caption{Comparison between the upper bounds of the bipartite EE obtained by the ground-state degeneracy of the auxiliary Hamiltonian (green dots) and the computed maximum bipartite EEs of all the scar states (blue dots), Eq.~\eqref{eq:AKLT_scars}, in the AKLT model.}
\label{fig:AKLT}
\end{figure}

To benchmark the effectiveness of our EE bounds, we numerically calculate the exact values of the bipartite EEs of all the scar states. Since our result bounds all the $\ket{\psi_{n}}$, we compare our bound with the maximum value of the numerically calculated values. As shown in Fig.~\ref{fig:AKLT}, the numerical EE maximum is approaching $(\log L_{A})/2$, and our bound is asymptotically $\log L_{A}$ \footnote{Similar results are also obtained for scar states in the extended Fermi-Hubbard model, see Appendix F that also includes Refs. \cite{Mark:2020ep,Moudgalya:2020ep,Yang:1989PRL_Hubbard,Vafek:2017bv}.}. Further increasing the operator range to $k=4$ and $k=5$, we find that $D^{A}$ remains the same as that of $k=3$ in both cases and that our upper bound can not be improved anymore. A natural guess is that our EE bound using $k=3$ is already tight. The reason behind our conjecture is that our bound actually applies not only to pure states $(Q^\dagger)^n \ket{G}$, but also to linear superpositions of them and even to mixed states constructed from them.

\textbf{Example 2: Ground states of the toric code model.}
Our formalism can be easily applied to stabilizer states \cite{Gottesman:1997sc,Gottesman:1998to,Pachos:Book}. Let us take the ground state of the toric code model \cite{Kitaev:2006ai} as a concrete example. Consider a 2D square lattice with qubits living on \textit{links} and, say, with the PBC. The toric code model is defined by the following Hamiltonian,
\begin{equation}
    \hat{H}_{\rm TC}= -\sum_\text{vertices}
    \begin{tikzpicture}[baseline=-0.5ex, thick, bullet/.style={draw,circle,inner sep=1pt, fill=black}]
        \def\L{0.7}
        \draw (0, 0) -- (0, \L) node[pos=0.5, bullet] {} node[pos=0.5, left] {\small $X$};
        \draw (0, 0) -- (\L, 0) node[pos=0.5, bullet] {} node[pos=0.5, above] {\small $X$};
        \draw (0, 0) -- (0, -\L) node[pos=0.5, bullet] {} node[pos=0.5, right]{\small $X$};
        \draw (0, 0) -- (-\L, 0) node[pos=0.5, bullet] {} node[pos=0.5, below] {\small $X$};
    \end{tikzpicture}
    - \sum_\text{faces}
    \begin{tikzpicture}[baseline=-0.5ex, thick, bullet/.style={draw,circle,inner sep=1pt, fill=black}]
        \def\L{0.8}
        \draw (-\L/2, -\L/2) -- ++(\L, 0) node[pos=0.5, bullet] {} node[pos=0.5, below] {\small $Z$} -- ++(0, \L) node[pos=0.5, bullet] {} node[pos=0.5, right] {\small $Z$}
        -- ++(-\L, 0) node[pos=0.5, bullet] {} node[pos=0.5, above] {\small $Z$}
        -- ++(0, -\L) node[pos=0.5, bullet] {} node[pos=0.5, left] {\small $Z$};
    \end{tikzpicture}
    \, .
\end{equation}
Each term in the first (second) sum is a product of four Pauli-$X$ (Pauli-$Z$) operators around a vertex (face), and is dubbed a star (plaquette) term. All the terms in $H_{\rm TC}$ commute with each other, and the ground state of the Hamiltonian is a simultaneous eigenstate of all the star and plaquette terms with eigenvalue $+1$. Now consider a rectangular subsystem as shown in Fig.~\ref{fig:ToricCodeSubsystem}(a) or \ref{fig:ToricCodeSubsystem}(b), and let us try to evaluate the EE for such a ground state. Since each star or plaquette term is proportional to a projector up to a constant shift, we may skip the initial steps of our formalism and directly construct the auxiliary Hamiltonian $\hat{H}^A_{\rm aux}$ as the negative sum of all the star and plaquette terms that are completely inside $A$, and the EE is bounded from above by $\log D^A$. To compute $D^A$, notice that terms in $\hat{H}^A_{\rm aux}$ form an \textit{independent} set of stabilizers \cite{Gottesman:1997sc,Gottesman:1998to,Pachos:Book}, which means (i) they are proportional to tensor products of the identity operator $\hat{I}$ and Pauli operators $\hat{X},\hat{Y},\hat{Z}$, (ii) they square to the identity and mutually commute, and (iii) any product of a nonempty subset of those operators is \textit{not} proportional to the identity (independence). These three properties imply that specifying the eigenvalue of each stabilizer will reduce the Hilbert space dimension by a half. Therefore, $D^A=2^{n_A-n_s}$ where $n_A$ is the number of qubits in the subsystem, and $n_s$ is the number of stabilizers contained in $\hat{H}^A_{\rm aux}$. A simple counting shows $S_v\leq \log D^A=(|\partial A|-1)\log2$ where $|\partial A|$ is the perimeter of the rectangular. This upper bound turns out to be \textit{exact}. The first term proportional to $|\partial A|$ indicates the area law, and the subleading constant term is known as the topological entanglement entropy \cite{Kitaev:2006PRL_EE,Levin:2006PRL_EE} signifying topological order.
\begin{figure}[t]
	\centering
	\includegraphics[width=0.95\linewidth]{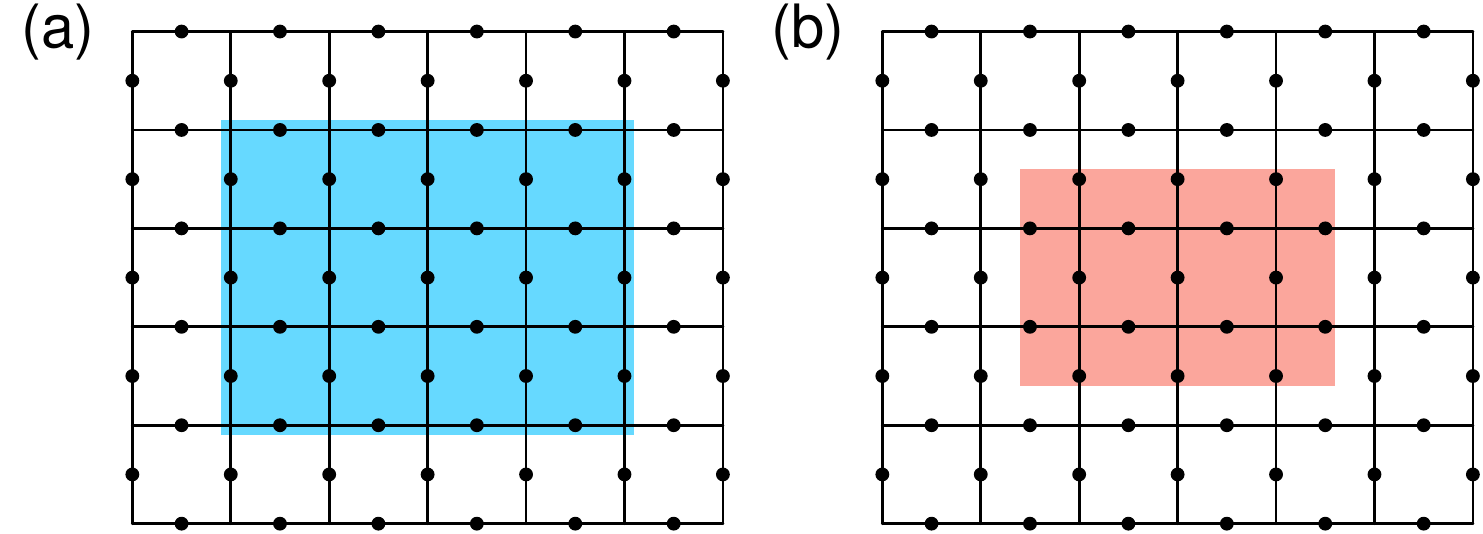}
	\caption{Two examples of a rectangular subsystem. The EE of the ground state of the toric code model takes the form $S_v=(|\partial A|-1)\log2$ in both cases where $|\partial A|$ is the perimeter of the shaded area.}
	\label{fig:ToricCodeSubsystem}
\end{figure}

\textit{\color{blue}Approximate Zeros.--} In previous discussions, we have demonstrated we can obtain useful upper bound on EE when the local correlation matrix contains exact zero eigenvalues. However, this protocol can easily fail. For one thing, the local correlation matrices of low entangled states, such as perturbed states of the ones considered in the previous section and many-body localized states, may contain only approximate zero eigenvalues.
For another, errors are inevitable in experimentally measured correlation matrices.
In this section, we extend previous results by providing an upper bound of EE when the correlation matrix contains approximate zero eigenvalues.

\begin{figure}[t]
	\centering
	\includegraphics[width=0.85\linewidth]{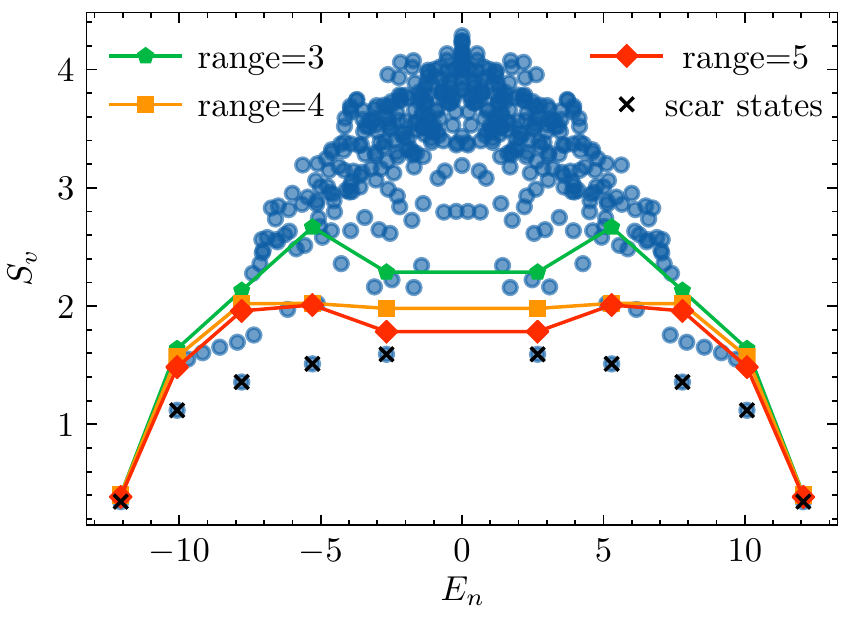}
	\caption{Comparision between our upper bounds (symbols shown in the legend) and the exact bipartite EEs of the scar states (black crosses) in the PXP model. The computation is perform for $L=20$ in the zero momentum and inversion even sector. For completeness, the exact EEs of all the energy eigenstates are shown. }
	\label{fig:PXP}
\end{figure}

Statements 1 and 2 require minor revision: instead of focusing only on zero eigenvalues, we also include eigenvalues $\lambda_\alpha \leq c$ with a finite cutoff $c$. Then following \eqref{eq:Delta} and \eqref{eq:DeltaP}, we can define an additional set of projection operators $\hat{P}_i$. Except for including these extra projectors, the form of the auxiliary Hamiltonian \eqref{eq:Haux} remains the same. Now, Statement 3 needs to be modified.

\textbf{Statement 3$'$: from Auxiliary Hamiltonian to the EE bound.} Given the auxiliary Hamiltonian $\hat{H}^A_\text{aux}$, we first compute its energy $E^A_\text{aux}$ of the state $|\psi\rangle$ as $E^A_\text{aux}=\text{Tr}_A (\hat{H}^A_\text{aux}\hat{\rho}_A) = \langle\psi|\hat{H}^A_\text{aux}|\psi\rangle$. We then consider the thermal reduced density matrix
\begin{equation}
    \hat{\rho}^A_\text{aux}(\beta)=\exp(-\beta \hat{H}^A_\text{aux})/\mathcal{Z}^A_\text{aux}(\beta).
\end{equation}
where $\mathcal{Z}^A_\text{aux}$ is the partition function and determine the inverse temperature $\beta^{*}$ by matching the energy $\text{Tr}_\text{A}(\hat{H}^A_\text{aux}\hat{\rho}^A_\text{aux}(\beta^*))=E^A_\text{aux}$. Since for all density matrices with the same energy, the thermal density matrix maximizes the entropy \cite{Sakurai:Book}, the von Neumann EE of $\ket{\psi}$ is upper bounded by the thermal entropy $S_{\text{th}}(\beta^*)=-\text{Tr}_\text{A}(\hat{\rho}^A_\text{aux}(\beta^*)\log \hat{\rho}^A_\text{aux}(\beta^*))$.

When we only include projectors of exact zeros in the auxiliary Hamiltonian, the statement 3' can be reduced to the statement 3. In this case, we have $E^A_\text{aux}=0$ and $\beta^*=\infty$, and consequently $S_{\text{th}}(\infty)=\log D^A$.

\textbf{Example 3: Scar states in the PXP model.}
We apply our extended protocol to the scar states in the PXP model under the PBC \cite{Bernien:2017pm,Turner:2018wg,Turner:2018qs}. We choose range-$k$ operator basis $\hat{L}_{i}$ as
\begin{equation}
    \hat{L}_{i} = \lambda_{a_1}^{(1)} \otimes \lambda_{a_2}^{(2)} \otimes...\otimes \lambda_{a_k}^{(k)} \, ,
\end{equation}
with $\{a_k\}=\{0,1,2,3\}$ being the indices of the identity matrix and Pauli matrices $\{\sigma_x,\sigma_y,\sigma_z\}$ (the identity operator where all $a_{k}$ are zero is excluded). In our numerical study, we choose the system size $L=20$ and $k \in \{3,4,5\}$. Due to the Hilbert space restriction, the correlation matrices for both scar states and thermal eigenstates have the same number of trivial zero eigenvalues, but those for scar states have more approximate zeros.  We only include the $\hat{\Delta}_{\alpha}$ operators with a certain cutoff $\lambda_{\alpha} < c$ in the auxiliary Hamiltonian, and at the same time enforce the Hilbert space restriction, effectively setting $\tilde{c}_{\alpha} = \infty$ for the $\hat{\Delta}_{\alpha}$ operators of trivial zeros. We then follow our protocol to compute the upper bounds \footnote{See Appendix G for details}.
As shown in Fig.~\ref{fig:PXP}, for $k=4,5$, the upper bounds of EE from the thermal entropy $S(\beta^*)$ lead to a well separation between thermal states and scar states.

\textit{\color{blue}Summary.--} In summary, we have developed a protocol that connects correlation and entanglement in a quantitative way, and tested it with exotic quantum states.
Our method can be readily applied to experiments as an economic way of bounding the entanglement entropy of a quantum state.
All one need to measure are local physical observables that give rise to the matrix elements of the local correlation matrices \footnote{See Appendix H for an experimental protocol in cold atom systems, which also includes Refs.~\cite{Bakr:2009aq,Sherson:2010sa,Weitenberg:2011ss,Cheuk:2015qg,Parsons:2015sr,Gross:2021qg,Wang:2015ca,Xia:2015rb,Wang:2016sq}}.
Different from methods that focus directly on the reduced density matrix \cite{Ohliger:2013ea,Lanyon:2017eo,Dalmonte:2018qs,Kokail:2021dc}, our method is divide-and-conquer in nature.
This means our measurement cost to bound the EE of a subsystem $A$ only scales \textit{linearly} with the subsystem size $L_{\mathrm{A}}$.
Moreover, for the exact zero case of eigenvalues of local correlation matrices, the numerical effort can be significantly reduced by our special algorithm. For the case of approximate zeros, however, using thermal entropy to bound entanglement entropy is costly and deserves further studies.
Still, in this Noisy Intermediate-Scale Quantum era \cite{Preskill:2018gt}, we envision our protocol to be particular useful.

\textit{Acknowledgements.} We especially thank Hui Zhai for many invaluable discussions. We are grateful to Xiao-Liang Qi for his interest and an insightful discussion. We also thank the referees for their many comments and suggestions that increase the clarity of this work. S. L. is supported by the Gordon and Betty Moore Foundation under Grant No. GBMF8690 and the National Science Foundation under Grant No. NSF PHY-1748958. P. Z. acknowledges support from the Walter Burke Institute for Theoretical Physics at Caltech.

\titleformat{\section}{\centering \bfseries}{Appendix \thesection:}{0.5em}{}
\appendix

\section{Illustration of Our Protocol of Exact Zeros Using $2$-local Operator Basis}
We shall further elucidate our protocol of bounding the entanglement entropy of a subsystem $A$ using $k=2$-local operator basis $\{\hat{L}_{i}\}$ here. To facilitate understanding, we use a concrete example depicted in Fig.~[\ref{fig:protocol}].
The subsystem $A$ is shaded red, say with sites $2, 3, 4, 5$, and of size $L_{A}=4$.
Our protocol is first to measure the correlation matrix on sites $2$ and $3$ (dotted blacks in the second row of Fig.~[\ref{fig:protocol}), collect the $\hat{\Delta}$ operators (equivalently projection operators) following statements 1 and 2 in the main text.
We then move to sites $3, 4$ to repeat the above procedure, and then to sites $4, 5$ for the same task.
After exhausting all these continuous $k=2$ subsystems inside $A$, we collect all those $\hat{\Delta}_{\alpha}$ to construct the auxiliary Hamiltonian $\hat{H}_{\text{aux}}^{A}$ following equation (6) in the main text.
The final step is to utilize the GSD-EE theorem to calculate the ground-state degeneracy $D^{A}$ that yields an upper bound of $\log D^{A}$ of the entanglement entropy of $A$.
\begin{figure}[h]
	\centering
	\includegraphics[width=.95\linewidth]{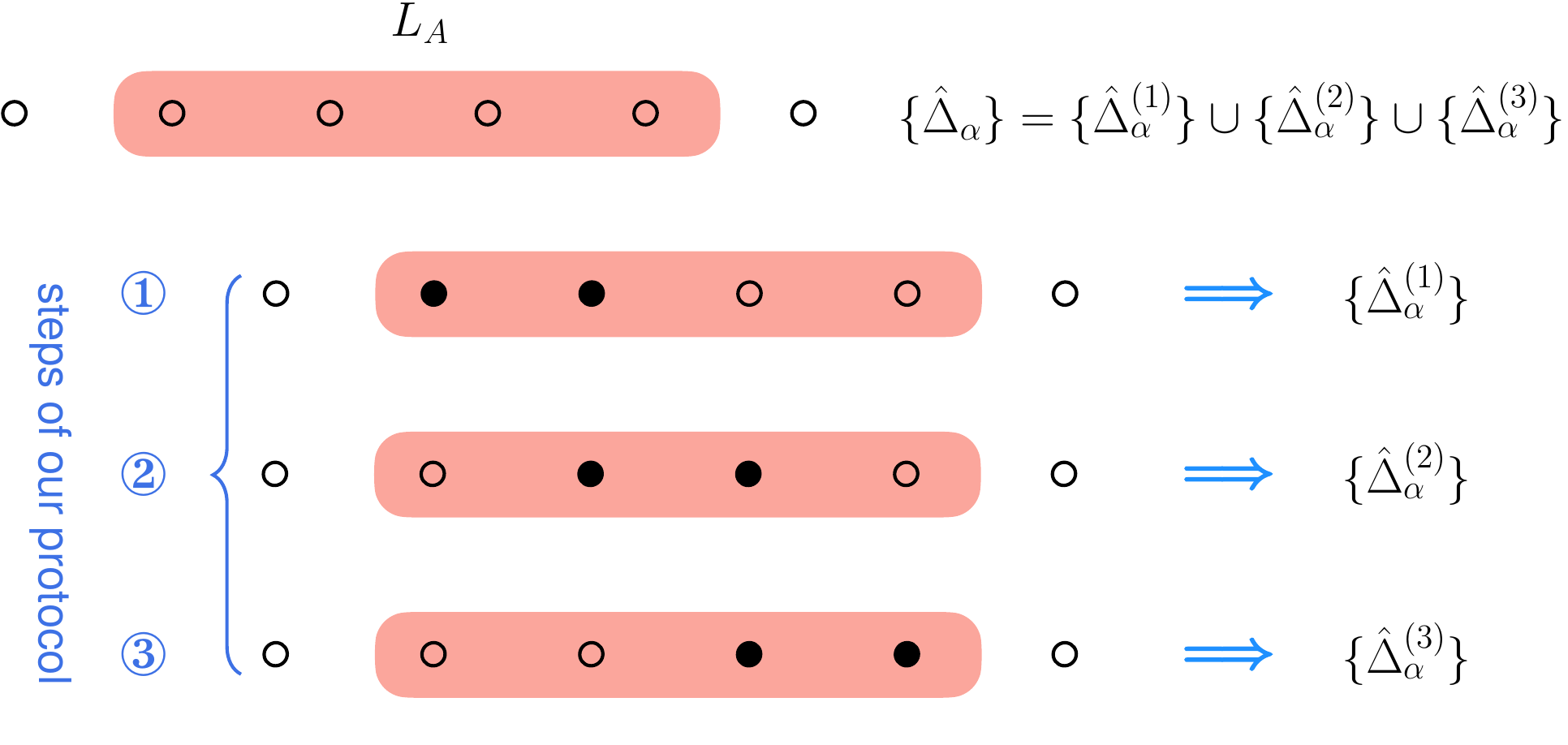}
	\caption{Illustration of steps of our protocol to bound entanglement entropy (EE) of subsystem $A$ (shaded red) using 2-local operator basis $\{\hat{L}_{i}\}$. For every consecutive two sites (dotted black) in $A$, we use a set of 2-local operator basis $\{\hat{L}_{i}\}$ acting on these two sites to construct the correlation matrix and determine the corresponding positive semidefinite operators $\{\hat{\Delta}_{\alpha}^{(n)} \}$ where $n$ labels the choices of these two sites. The $\hat{\Delta}_{\alpha}$ operators that enters the auxiliary Hamiltonian comes from
    the union of these $\{\hat{\Delta}_{\alpha}^{(n)}\}$.}
	\label{fig:protocol}
\end{figure}

\section{Equivalence between Different Choices of Eigenvectors $\{\mathbf{e}_{\alpha}\}$ with Degeneracy}

We shall show in the degenerate case where $\operatorname{dim}(\operatorname{Ker} (\mathcal{M})) = n > 1$ the final result of $D^{A}$ does not depend on the choice of eigenvectors of zero eigenvalues (our \textbf{Lemma 1} below). Let us denote the unitary matrix connecting two different sets of zero-eigenvalue basis $\{\mathbf{e}_{\alpha}\}$ and $\{\mathbf{\tilde{e}}_{\alpha}\}$ as $U$,
\begin{equation}
	\begin{pmatrix}
		\mathbf{\tilde{e}}_{1} \\
		\mathbf{\tilde{e}}_{2} \\
		\vdots \\
		\mathbf{\tilde{e}}_{n}
	\end{pmatrix}
	=
	\begin{pmatrix}
		U_{11} & U_{12} & \cdots & U_{1n} \\
		U_{21} & U_{22} & \cdots & U_{2n} \\
		\vdots \\
		U_{n1} & U_{n2} & \cdots & U_{nn}
	\end{pmatrix}
	\begin{pmatrix}
		\mathbf{e}_{1} \\
		\mathbf{e}_{2} \\
		\vdots \\
		\mathbf{e}_{n}
	\end{pmatrix} \, ,
\end{equation}
and the state from which the correlation matrix is constructed as $|\psi\rangle$. From
\begin{equation}
    \hat{O}_{\alpha} \ket{\psi} = \xi_{\alpha} \ket{\psi} \, ,
\end{equation}
we can easily write down the eigenoperator equation for $\hat{\widetilde{O}}_{\alpha}$,
\begin{equation*}
    \hat{\widetilde{O}}_{\alpha} \ket{\psi} = \sum_{\beta} U_{\alpha \beta} \hat{O}_{\beta} \ket{\psi} = \sum_{\beta} U_{\alpha \beta} \xi_{\beta} \ket{\psi} \equiv \widetilde{\xi}_{\alpha} \ket{\psi} \, ,
\end{equation*}
where we have identified $\widetilde{\xi}_{\alpha} = \sum_{\beta} U_{\alpha \beta} \xi_{\beta}$. The corresponding positive semidefinite operator $\hat{\widetilde{\Delta}}_{\alpha}$ for the tilde version then reads
\begin{eqnarray}
	\hat{\widetilde{\Delta}}_{\alpha} &=& (\hat{\widetilde{O}}_{\alpha}^{\dagger} - \tilde{\xi}^{*}_{\alpha} \hat{I}) ( \hat{\widetilde{O}}_{\alpha} - \tilde{\xi}_{\alpha} \hat{I}) \nonumber \\
     & = & \sum_{\beta, \gamma} U_{\alpha \beta} U_{\alpha \gamma}^{*} (\hat{O}_{\gamma}^{\dagger} - \xi^{*}_{\gamma} \hat{I}) ( \hat{O}_{\beta} - \xi_{\beta} \hat{I})  \label{eq:tilde_Delta} \, .
\end{eqnarray}

\begin{theorem}
	The value of ground-state degeneracy of the auxiliary Hamiltonian, $D^{A}$, is the same for different basis choices $\{\mathbf{e}_{\alpha}\}$ and $\{\tilde{\mathbf{e}}_{\alpha}\}$.
\end{theorem}

\begin{proof}
	It suffices to show if a ground state $\ket{G}$ of the auxiliary Hamiltonian $\hat{H}^{A}_{\text{aux}}$ satisfy $\hat{\widetilde{\Delta}}_{\alpha} \ket{G} = 0, \; \forall \alpha$, then $\hat{\Delta}_{\alpha} \ket{G} = 0, \; \forall \alpha$ since we can reverse the role $\{\mathbf{e}_{\alpha}\}$ and $\{\tilde{\mathbf{e}}_{\alpha}\}$ and repeat the the whole process. From our condition $\hat{\widetilde{\Delta}}_{\alpha} \ket{G} = 0, \; \forall \alpha$ and Eq.~\eqref{eq:tilde_Delta}, we have
	\begin{equation}
		\bra{G} \sum_{\beta, \gamma} U_{\alpha \beta} U_{\alpha \gamma}^{*} (\hat{O}_{\gamma}^{\dagger} - \xi^{*}_{\gamma} \hat{I}) ( \hat{O}_{\beta} - \xi_{\beta} \hat{I}) \ket{G} = 0 \, .
	\end{equation}
	Summing the $\alpha$ index and noting $\sum_{\alpha} U_{\alpha \beta} U_{\alpha \gamma}^{*} = \delta_{\beta \gamma}$ (delta function), we have
	\begin{equation*}
		\bra{G} \sum_{\beta}  (\hat{O}_{\gamma}^{\dagger} - \xi^{*}_{\beta} \hat{I}) ( \hat{O}_{\beta} - \xi_{\beta} \hat{I}) \ket{G} = \bra{G} \sum_{\alpha} \hat{\Delta}_{\alpha} \ket{G} = 0  \, ,
	\end{equation*}
	where we have changed the dummy summation index from $\beta$ to $\alpha$. Since $\hat{\Delta}_{\alpha}$ is positive semidefinite, the above equation implies
	\begin{equation}
		\hat{\Delta}_{\alpha} \ket{G} = 0  \,, \quad \forall \alpha \, .
	\end{equation}
	Therefore, $D^{A}$ does not depend on the choice of eigenvector basis in the degenerate space of zero eigenvalues.
\end{proof}

\section{EE Bounds of Product States }
In this section, we take product states as an example for demonstrating the effectiveness of our GSD-EE theorem. We shall prove the local correlation matrix of 1-local operator basis $\{\hat{L}_{i}\}$ is enough to give an EE upper bound of zero for produce state.

For definiteness, consider a spin-1/2 system and denote the local bases on a site as $\ket{0}$ and $\ket{1}$. Without loss of generality, we can write the product state as
\begin{equation}
    \ket{\psi} = \prod_{i} \otimes \ket{0}_{i} = \ket{0 0 \cdots 0}  \, ,
\end{equation}
since one can always rotate local basis on every site freely. Now, choose the nontrivial 1-local operator basis as the standard Pauli operators $\hat{X}, \hat{Y}$, and $\hat{Z}$.
The correlation matrix, defined in Eq.~(1) in our manuscript, is thus
\begin{equation}
    \mathcal{M}_{ij} = \bra{\psi} \hat{\sigma}_{i}  \hat{\sigma}_{j} \ket{\psi} - \bra{\psi} \hat{\sigma}_{i}^{\dagger} \ket{\psi} \bra{\psi} \hat{\sigma}_{j} \ket{\psi}
\end{equation}
where we have used the shorthand notation $\hat{\bm{\sigma}} \equiv (\hat{X}, \hat{Y}, \hat{Z})$. The explicit form of $\mathcal{M}$ can be easily written down
\begin{equation}
\mathcal{M} =
    \begin{pmatrix}
    1 & -i & 0 \\
    i &  1 & 0 \\
    0 &  0 & 0
    \end{pmatrix} \, .
\end{equation}
The eigenvalues of $\mathcal{M}$ are $0$, $0$, and $2$, and two linear independent eigenvectors of $\mathcal{M}$ with zero eigenvalues can be taken as $\mathbf{e}_{1} = (1/\sqrt{2}, -i/\sqrt{2}, 0)^{T}$ and $\mathbf{e}_{2} = (0, 0, 1)^{T}$. The corresponding two eigenoperators are $\hat{O}_{1} = (\hat{X} - i \hat{Y})/\sqrt{2}$ with $\xi_{1} = 0$ and $\hat{O}_{2} = \hat{Z}$ with $\xi_{2} = -1$.
The associated positive semidefinite operators, defined in Eq.~(4) in the main text, $\hat{\Delta}_{1}$ and $\hat{\Delta}_{2}$ are
\begin{equation}
\begin{aligned}
    \hat{\Delta}_{1} &= (\hat{X} + i\hat{Y})(\hat{X} - i\hat{Y})/2 = \hat{I} + \hat{Z} = 2 \ket{1} \bra{1} \, , \\
    \hat{\Delta}_{2} &= (\hat{Z} + \hat{I}) (\hat{Z} + \hat{I}) = 2(\hat{Z} + \hat{I}) = 4 \proj{1} \, .
\end{aligned}
\end{equation}
This means the auxiliary Hamiltonian $\hat{H}^{A}_{\text{aux}}$, defined in Eq.~(6) in the main text, takes the form
\begin{equation}
    \hat{H}^{A}_{\text{aux}} = \sum_{i \in A} c_{i} \ket{1}_{i} \bra{1}_{i} \, , \quad c_{i} > 0 \, ,
\end{equation}
where $i$ labels the sites in subsystem A. Clearly, \emph{regardless of the value of $c_{i}$ as long as $c_{i} > 0$}, the ground state of $\hat{H}^{A}_{\text{aux}}$ under the open boundary condition is unique: $\ket{G} = \prod_{i \in A} \otimes \ket{0}_{i}$. From our GSD-EE theorem, the von Neumann entropy $S_{v}$ is upper bounded by $\log 1 = 0$. On the other hand, $S_{v} \geq 0 $, and, therefore, our GSD-EE theorem yields the exact result $S_{v} = 0$ for product states.\\

A remark is in order here. In practice we do not know the rotated basis, $\ket{0}$ and $\ket{1}$, and hence do not know $\hat{X}, \hat{Y}$, and $\hat{Z}$ defined above. However, the choice of bases $\{ \hat{L}_{i} \}$ is immaterial. Indeed, if one repeats the above process for a general product state,
\begin{equation}
    \ket{\psi} = \prod_{i} \otimes \ket{\phi_{i}} = \prod_{i} \otimes (\alpha_{i} \ket{0} + \beta_{i} \ket{1}) \, ,
\end{equation}
one will again find two zero eigenvalues of the correlation matrix. The two eigenvectors, up to some normalization constants, can be taken as
\begin{equation}
\begin{aligned}
    \mathbf{e}_{1} & = \big( \dfrac{\alpha \beta^{*} + \alpha^{*} \beta}{|\beta|^{2} - |\alpha|^{2}}, \dfrac{i\alpha^{*} \beta - i\alpha \beta^{*}}{|\beta|^{2} - |\alpha|^{2}}, 1   \big)^{T} \, , \\
    \mathbf{e}_{2} & = \big( \dfrac{\alpha^{2} - \beta^{2}}{2\alpha \beta}, -i \dfrac{\alpha^{2} - \beta^{2}}{2\alpha \beta}, 1 \big)^{T} \, .
\end{aligned}
\end{equation}
The eigenvalues $\xi_{1}$ and $\xi_{2}$ for $\hat{O}_{1}$ and $\hat{O}_{2}$ are $\xi_{1} = 1/(|\beta|^{2} - |\alpha|^{2})$ and $\xi_{2} = 0$, respectively. The two $\hat{\Delta}$ operators, up to some constant normalization factors, constructed for site $i$ are
\begin{equation}
\hat{\Delta}_{1} =
\hat{\Delta}_{2} =
\begin{pmatrix}
|\alpha_{i}|^{2}  & -\alpha_{i}^{*} \beta_{i} \\
-\alpha_{i} \beta_{i}^{*} & |\beta_{i}|^{2}
\end{pmatrix}
= \proj{\phi_{i}^{\perp}} \, ,
\end{equation}
where $\ket{\phi_{i}^{\perp}} = \beta_{i}^{*} \ket{0} - \alpha_{i}^{*} \ket{1}$ is orthogonal (zero inner product) to $\ket{\phi_{i}}$. The resulting auxiliary Hamiltonian then reads (up to some normalization constants)
\begin{equation}
    \hat{H}^{A}_{\text{aux}} = \sum_{i} c_{i} \proj{\phi_{i}^{\perp}} \, , \quad c_{i} > 0 \, ,
\end{equation}
and its ground state is unique with $D^{A} =1$.\\

Having showed the special case of spin-1/2 systems where the local Hilbert space dimension $q$ equals two, we shall continue to give a formal constructive proof for the case of general $q$. Denote the local Hilbert basis as $\ket{i}, i=0, 1, \cdots, q-1$ and still, without loss of generality, the general form a product state $\ket{\psi}$ is taken to be
\begin{equation} \label{eq:prod_0}
	\ket{\psi} = \prod_{i} \otimes \ket{0}_{i} = \ket{0 0 \cdots 0}  \, ,
\end{equation}
as before. And we shall take the operator basis $\{\hat{L}_{i}\}$ to be the set of complete 1-local operators (still the trivial identity operator can be neglected as before). Then, a general 1-body eigenoperator $\hat{O}$ where $\hat{O} \ket{\psi} = \xi \ket{\psi}$ (they exist, for example we can construct $\ket{0}\bra{1}$ as one such operator) can be formally written as
\begin{equation}
	\hat{O} = \xi \proj{0} + \sum_{i=1}^{q-1} \ketbra{\phi_{i}}{i} \, ,
\end{equation}
where $\ket{\phi_{i}} = \hat{O} \ket{i}$. The positive semidefinite operator $\hat{\Delta}$ then reads
\begin{equation}
	\hat{\Delta} = (\hat{O}^{\dagger} - \xi^{*} \hat{I}) ( \hat{O} - \xi \hat{I}) = \sum_{i, j=1}^{q-1} C_{ji} \ketbra{j}{i}
\end{equation}
where
\begin{equation}
	C_{ji} = \left( \bra{\phi_{j}} - \xi^{*} \bra{j} \right) \left( \ket{\phi_{i}} - \xi \ket{i} \right)
\end{equation}
is positive semidefinite (can be proven by showing $\sum_{i,j} x_{j}^{*} C_{ji} x_{i} \geq 0$ since it is the squared norm of $x_{i}(\ket{\phi_{i}} - \xi \ket{i})$) as expected. Therefore, we can diagonalize $C_{ij}$ and bring $\hat{\Delta}$ into diagonal form
\begin{equation}
    \hat{\Delta} = \sum_{i, j=1}^{q-1} C_{ji} \ketbra{j}{i} = \sum_{i} a_{i} \proj{i} \equiv \sum_{i} a_{i} \hat{P}_{i} \, ,
\end{equation}
where $a_{i} \geq 0$ are the eigenvalues of matrix $C_{ij}$. The important thing to note is that $\hat{P}_{i} \ket{0} = \ket{\tilde{i}} \braket{\tilde{i} | 0} = 0$ since $\ket{\tilde{i}}$ is a linear combination of $\ket{i}$ with $i=1, 2, \cdots, q-1$. Moreover, by taking a set of states $|\phi_j \rangle$ where the matrix $C_{ij}$ has rank $q-1$ (for example, setting $\ket{\phi_i} = (\xi + 1) \ket{i}$ makes $C$ an identity matrix), $|0\rangle$ becomes the only state in the kernel of $\hat{\Delta}$. As a result, the ground state of $\hat{H}^{A}_{\text{aux}}$ under the open boundary condition is unique as in the spin-1/2 case: $\ket{G} = \prod_{i \in A} \otimes \ket{0}_{i}$, and our upper bound of the von Neumann entropy is exact, $S_{v} = \log 1 = 0$.

\section{Special Algorithm for the Ground-State Degeneracy of the Auxiliary Hamiltonian} \label{sup:Algorithm}
Our GSD-EE theorem states that the EE of a state $\ket{\psi}$ on a subsystem of size $L_{A}$ is upper bounded by the ground-state degeneracy $D^{A}$ of the auxiliary Hamiltonian $\hat{H}^{A}_{\text{aux}}$ under the open boundary condition. Because of the exponential increase of the Hilbert space dimension, the maximum size $L_{A}$ can be reached by directly attacking this problem using the exact diagonalization method is quite limited. We now explain a special algorithm that overcomes the difficulty. For the simplicity of notations, consider the concrete example of AKLT scar states mentioned in the main text. Here, due to the translation symmetry, we obtain the same set of projectors from the local correlation matrix of any three consecutive sites.
Let $\ket{G_{m}^{(L_{A})}}$ denote the orthonormal ground states of the auxiliary Hamiltonian of system size $L_{A}$ with $m=1, 2, \cdots, D^{A}$ as the index, and let $\hat{P}^{(L)}_{i}$ denote the $i$-th projector acting on the sites $L-2, L-1$ and $L$, so that
\begin{equation}
	\hat{H}^{A}_{\text{aux}} = \hat{H}^{A'}_{\text{aux}} + \sum_{i} c_{i} \hat{P}^{(L_{A})}_{i} \equiv \hat{H}^{A'}_{\text{aux}} + \hat{H}^{(L_{A})}
\end{equation}
where $A'$ is the subsystem of sites the $1, 2, \cdots, L_{A}-1$, and $\hat{H}^{(L_{A})}$ is the part that acts only on the three sites $L_{A}-2, L_{A}-1$ and $L_{A}$. Since $\ket{G^{(L_{A})}_{m}}$ is annihilated by every projector acting on the sites $1, 2, \cdots, L_{A} -1$, we can expand it using the ground states $\ket{G^{(L_{A}-1)}_{k}}$  of the auxiliary Hamiltonian with system size $L_{A} - 1$ and the basis $\ket{e_{n}}$ on the $L_{A}$-th site,
\begin{equation} \label{eq:recursive}
	\ket{G_{m}^{(L_{A})}} = \sum_{n, k} C_{kn, m}^{(L_{A})} \ket{G^{(L_{A}-1)}_{k}} \otimes  \ket{e_{n}}
\end{equation}
where $n$ and $k$ are the indices of the corresponding orthonormal states and $C_{kn, m}^{(L_{A})}$ is the unknown coefficient. These coefficients and the degeneracy $D^{A}$ can be obtained by diagonalizing $\hat{H}^{A}_{\text{aux}}$ or equivalently $\hat{H}^{(L_{A})}$ under the basis $\ket{G^{(L_{A}-1)}_{k}} \otimes  \ket{e_{n}}$. Moreover, we do not need to store the memory-consuming part $\ket{G^{(L_{A}-1)}_{k}}$ in this process. Recursively using \eqref{eq:recursive}, we can expand $\ket{G^{(L_{A}-1)}_{k}}$ as
\begin{equation*}
    \sum_{l, n} C_{ln, k}^{(L_{A}-1)} \left( \sum_{i,j} C_{ij, l}^{(L_{A}-2)}  \ket{G_{i}^{(L_A-3)}} \otimes \ket{e_{j}} \right) \otimes \ket{e_{n}} \, .
\end{equation*}
Since $\hat{H}^{(L_{A})}$ acts trivially on $\ket{G_{i}^{(L_A-3)}}$, we just need to store the coefficients $C_{ij, l}^{(L_{A}-2)}$ and $C_{ln, m}^{(L_{A}-1)}$ to calculate the matrix elements of $\hat{H}^{(L_{A})}$ in the basis $\ket{G^{(L_{A}-1)}_{k}} \otimes  \ket{e_{n}}$, diagonalize the resulting matrix, and obtain the ground state degeneracy $D^{A}$ as well as $C_{kn, m}^{(L_{A})}$. Since this recursive method is memory-cheap and we just need to diagonalize a matrix of dimension $D^{A'} \times q$ (where $q$ is the Hilbert space dimension of a single site) to get $D^{A}$, calculating the ground state degeneracy of system size $L_{A}$ of a few hundreds is very easy.

\section{Projector Space of Scar States in the AKLT Model} \label{sup:AKLT}
To employ our GSD-EE theorem to upper bound the bipartite entanglement entropies of all the scar states,
\begin{equation} \label{eq:AKLT_scar}
	\ket{\psi_{n}} = (\hat{Q}^{\dagger})^{n} \ket{G},
\end{equation}
we need to find a set of projection operators $\{\hat{P}_{i}\}$ such that each of them annihilates every $\ket{\psi_{n}}$,
\begin{equation}
	\hat{P}_{i} \ket{\psi_{n}} = 0 \, , \quad \forall \, i, n \, .
\end{equation}
Let $\hat{H}_{\text{aux}, n}^{A}$ denotes an auxiliary Hamiltonian of $\ket{\psi_{n}}$, let  $\hat{P}_{i}^{(n)} = \ket{i^{(n)}} \bra{i^{(n)}}$ denotes the projectors entering $\hat{H}_{\text{aux}, n}^{A}$. Since, as implied by \textbf{Lemma 1}, the kernel (ground-state manifold) of the auxiliary Hamiltonian is \textit{invariant} for fixed operator space $\mathcal{V} = \operatorname{Span} \{ \hat{L}_{i}\}$, the complement of the kernel, the state space $V^{(n)} \equiv \operatorname{Span} \{  \ket{i^{(n)}}\}$, is also invariant. Therefore, the direct way to identify $\{ \hat{P}_{i} \}$ annihilating all the scar states is to first calculate the common intersection $V_{\cap} \equiv \bigcap\limits_{n} \, V^{(n)} $ and then form $\{ \proj{i} \}$ using a set of bases $\{\ket{i}\}$ in $V_{\cap}$. This is, however, computationally too expensive, for it requires calculating $V^{(n)}$ for every $\ket{\psi_{n}}$.  Fortunately, making use of the special structure of the AKLT Hamiltonian and the ground state $\ket{G}$, we typically only need to deal with one $\ket{\psi_{n}}$, thus significantly reducing the computational cost. First note that the AKLT Hamiltonian under the periodic boundary condition (PBC) is a sum of projectors,
\begin{equation}
	\hat{H}_{\text{AKLT}} = \sum_{i=1}^{L} \hat{P}_{i,i+1}^{(2)} \, ,
\end{equation}
where $\hat{P}_{i, i+1}^{(2)}$ is the projection operator of two spin-1's at site $i$ and $i+1$ onto total spin-2, and that the ground state $\ket{G}$ is annihilated by every  $\hat{P}_{i, i+1}^{(2)}$,
\begin{equation}
	\hat{P}^{(2)}_{i,i+1} \ket{G} = 0 \, , \quad \forall \; i \, .
\end{equation}
If we partite the system into two parts, one consisting of $k$ consectutive sites and the other part containing the rest $L-k$ sites as shown in Fig.~[\ref{fig:decomposition}], and denote the part of the Hamiltonian acting only in the $k$-site region as $\hat{H}_{k}$, i.e.
\begin{equation}
	\hat{H}_{k} = \sum_{i=1}^{k-1} \hat{P}_{i,i+1}^{(2)} \, ,
\end{equation}
it follows that the ground state is annihilated by $\hat{H}_{k}$,
\begin{equation} \label{eq:AKLT_k}
	\hat{H}_{k} \ket{G} = 0 \, .
\end{equation}
\begin{figure}[htbp]
	\centering
	\includegraphics[width=0.8\linewidth]{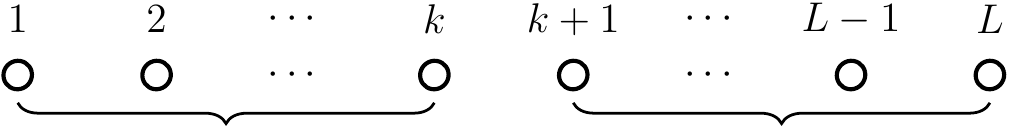}
	\caption{Partition of the system into $k$-site region and $(L-k)$-site region.}
	\label{fig:decomposition}
\end{figure}
Since $\hat{H}_{k}$ is nothing but the AKLT Hamiltonian of the $k$-site region under the open boundary condition (OBC), Eq.~\eqref{eq:AKLT_k} implies $\ket{G}$ can be decomposed in the following way,
\begin{equation} \label{eq:G_decomp}
	\ket{G} = \sum_{\eta=1}^{4} \ket{G_{k, \eta}} \otimes \ket{\Psi_{\eta}}
\end{equation}
where $\ket{{G}_{k,\eta}}$ is the ground state of the $k$-site AKLT Hamiltonian under the OBC with $\eta=1,2,3,4$ as labels of the four-fold degeneracy, and $\ket{\Psi_{\eta}}$ is a state vector in the Hilbert space of the rest $L-k$ sites. Similarly, we can
write $\hat{Q}^{\dagger}$ as a sum of two terms, denoted by $\hat{Q}_{k}$ and $\hat{Q}_{L-k}^{\dagger}$, that act only in the $k$- and $(L-k)$-site regions respectively,
\begin{equation*} \label{eq:Q_decomp}
    \hat{Q}^{\dagger} = \sum_{i=1}^{k} (-1)^{i} \left( \hat{S}_{i}^{+} \right)^{2} + \sum_{i=k+1}^{L} (-1)^{i} \left( \hat{S}_{i}^{+} \right)^{2}  \equiv \hat{Q}_{k}^{\dagger} + \hat{Q}_{L-k}^{\dagger} \, .
\end{equation*}
The key observation is that, using the decomposition of $\ket{G}$ and $\hat{Q}^{\dagger}$, we can decompose a general scar state, Eq.~\eqref{eq:AKLT_scar}, as
\begin{eqnarray} \label{eq:scar_decomp}
\ket{\psi_{n}} &= & \left( \hat{Q}_{k}^{\dagger} + \hat{Q}_{L-k}^{\dagger} \right)^{n} \ket{G} \nonumber  \\
      & = & \sum_{m=0}^{n} \sum_{\eta=1}^{4} (\hat{Q}_{k}^{\dagger})^{m} \ket{G_{k, \eta}} \otimes \ket{\tilde{\Psi}_{m, \eta}} \, ,
\end{eqnarray}
where $\ket{\tilde{\Psi}_{m, \eta}}$ is a state vector in the Hilbert space the $L-k$ sites. Since $\ket{G_{k, \eta}}$ only contain $z$-component of the total spin, $S_{z}$, greater or equal than $-1$ and each $\hat{Q}_{k}^{\dagger}$ raises $S_{z}$ by two \cite{Auerbach:Book}, the maximum value of $m$ we need to address is $\left\lfloor\frac{k+1}{2} \right\rfloor$ where $\left\lfloor x \right\rfloor$ denotes the maximum integer no bigger than $x$. Let $\{ \ket{\Phi_{i}} \}$ denotes a set of orthonormal bases where each basis has zero inner product with $(\hat{Q}_{k}^{\dagger})^{m} \ket{G_{k, \eta}}$ for every possible values of $m$ and $\eta$,
\begin{equation}
    \langle \Phi_{i} \vert (\hat{Q}_{k}^{\dagger})^{m} \ket{G_{k, \eta}}  = 0 \, , \quad \forall \; i, m, \eta \, .
\end{equation}
We can then easily construct a set of $k$-local projection operators $\{\proj{\Phi_{i}}\}$
and each $\proj{\Phi_{i}}$ annihilates all the scar states, Eq.~\eqref{eq:AKLT_scar}, for system size $L>k$. From construction, it is clear that the state space spanned by
$\{ \ket{\Phi_{i}} \}$, denoted as $V_{\Phi}$, is a subset of the common intersection $V_{\cap}$ we are looking for,
\begin{equation} \label{eq:VinV1}
	V_{\Phi} \equiv \operatorname{Span}\{\, \ket{\Phi_{1}}, \ket{\Phi_{2}}, \cdots \,\} \subseteq V_{\cap} \, .
\end{equation}
On the other hand, for a scar state $\ket{\psi_{n}}$ with fixed $ n \geq \lfloor \frac{k+1}{2} \rfloor$ of a \textit{large enough} system size $L$, the states $\ket{\tilde{\Psi}_{m,\eta}}$ in the decomposition of $\ket{\psi_{n}}$, Eq.~\eqref{eq:scar_decomp}, are in general linearly independent. This implies that for such a scar state $\ket{\psi_{n}}$,
\begin{equation*}
	\hat{P}_{i}^{(n)} \ket{\psi_{n}} = 0 \, , \; \Leftrightarrow \; \hat{P}_{i}^{(n)} (\hat{Q}_{k}^{\dagger})^{m} \ket{G_{k, \eta}}  = 0 \, , \; \forall \; m, \eta \, .
\end{equation*}
In other words, $V^{(n)} = V_{\Phi}$ for such $\ket{\psi_{n}}$, which further implies $V_{\cap} \subseteq V_{\Phi}$. Together with Eq.~\eqref{eq:VinV1}, we conclude
\begin{equation*}
    V_{\Phi} = V_{\cap} \, ,
\end{equation*}
and $V_{\cap}$ can be determined by such a scar state $\ket{\psi_{n}}$. \\

The above consideration holds generally for $k \geq 2$, and further simplification can happen for a given $k$. For example, for $k=3$, a simple calculation shows that
\begin{equation*}
	\operatorname{Span} \{ \, (\hat{Q}_{3}^{\dagger})^{2} \ket{G_{3, \eta}} \, \} \subset \operatorname{Span} \{ \, \ket{G_{3, \eta}}, \hat{Q}_{3}^{\dagger} \ket{G_{3, \eta}} \, \} \, .
\end{equation*}
This means that a $3$-local projector that annihilates $\ket{\psi_{1}} = \hat{Q}^{\dagger} \ket{G}$ will also annihilates all the other scar states $\ket{\psi_{n}}$ for all system sizes $L>3$. Moreover, our $3$-local operator basis
\begin{equation}
    \hat{L}_{i} = \lambda_{a}^{(1)} \otimes \lambda_{b}^{(2)} \otimes \lambda_{c}^{(3)} \, ,  \quad a, b, c =1, 2, \cdots, 8
\end{equation}
is general enough for the constructed state space $V = \operatorname{Span} \{ \ket{i}\}$ from $\{ \hat{P}_{i}\}$ to contain $V_{\Phi}$. Similar result holds both for 4-local operator basis
\begin{equation}
	\hat{L}_{i} = \lambda_{a}^{(1)} \otimes \lambda_{b}^{(2)} \otimes \lambda_{c}^{(3)} \otimes \lambda_{d}^{(4)} \, ,
\end{equation}
and 5-local operator basis
\begin{equation}
	\hat{L}_{i} = \lambda_{a}^{(1)} \otimes \lambda_{b}^{(2)} \otimes \lambda_{c}^{(3)} \otimes  \lambda_{d}^{(4)} \otimes \lambda_{e}^{(5)} \, ,
\end{equation}
where $a, b, c, d, e =1, 2, \cdots, 8$.

\section{EE Bounds of Scar States in the Fermi--Hubbard Model} \label{FermiHubbard}
The scar states in the 1D extended Fermi--Hubbard model for spin-$1/2$ fermions take a similar form to those of the AKLT model \cite{Mark:2020ep,Moudgalya:2020Hubbard}. They are known as the $\eta$-pairing states \cite{Yang:1989PRL_Hubbard} and take the form of
\begin{equation} \label{eq:Hubbard_scar}
    \ket{\psi_{n}} = (\hat{\eta}^{\dagger})^{n} \ket{0} \, ,
\end{equation}
where $\ket{0}$ is the vacuum state with no fermions and $\hat{\eta}^\dagger = \sum_{i} (-1)^{i} \hat{c}_{i\uparrow}^{\dagger} \hat{c}_{i\downarrow}^{\dagger}$ is the so-called $\eta$-pairing operator with $\hat{c}_{i\uparrow}^{\dagger}$ ($\hat{c}_{i\downarrow}^{\dagger}$) being the fermionic creation operator for spin up (down) at site $i$.
The procedure to construct the local correlation function parallels that of the AKLT model.
We choose our target scar state as the half-filled state $\ket{\psi}_{L/2}$ ($L$ is assumed to be even), and operators $\hat{L}_i$ as range-$3$ operators with a similar tensor product form, $\hat{L}_{i} = \lambda_{a}^{(1)} \otimes \lambda_{b}^{(2)} \otimes \lambda_{c}^{(3)}$ on three neighboring sites $1$, $2$ and $3$.
The Hilbert space is four-dimensional for each site, spanned by an unoccupied state, two singly occupied states with different spins, and one doubly occupied state.
Thus, here we take $\lambda$ as the fifteen generators of SU(4) group and the identity ($a, b, c = 1,\cdots,16$). The local correlation matrix is then calculated and the auxiliary Hamiltonian is constructed in a similar fashion.
We find ground state degeneracy $D^A$ scales with the subsystem size $L_\text{A}$ as $D^A=L_{A}+1$, meaning the EE of this state is bounded by the logarithm of the subsystem size.
For reasons identical to the AKLT case, this bound is also valid for all other scar states of $n \neq L/2$. Finally, we compare the exact value of EE with our bound of the scar state $\ket{\psi_{L/2}}$. The exact value of EE is analytically known \cite{Vafek:2017bv} to be asymptotically $S_{v} \propto (\log L_{A})/2$. Therefore, apart from a ratio of two, our bound gives the correct log-volume scaling of the EE.

\begin{figure}[htbp]
\centering
\includegraphics[width=0.8\linewidth]{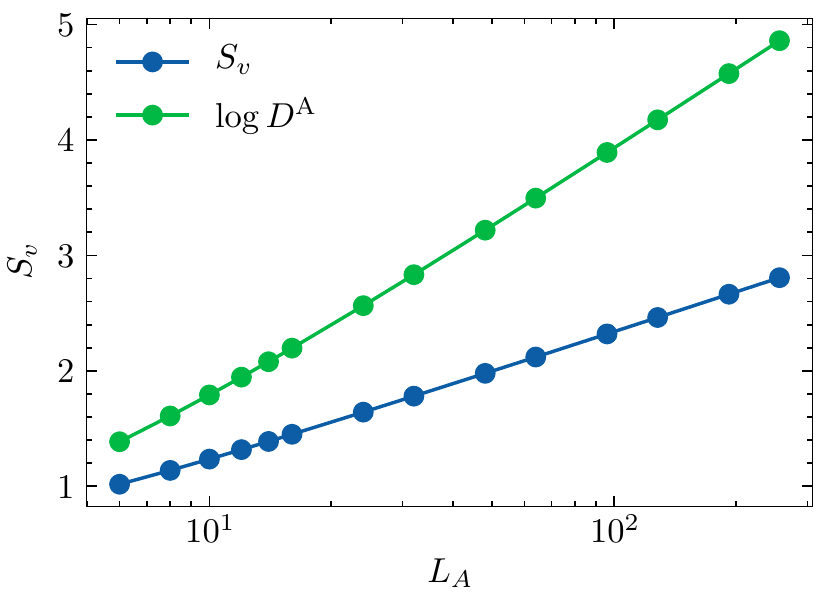}
\caption{Comparison between the EE bounds (green dots) from the ground state degeneracies of the auxiliary Hamiltonian $\hat{H}^A_{\text{aux}}$ and the exact values (blue dots) for $|\psi_{L/2} \rangle$ for various subsystem size $L_\text{A}$. }
\label{fig:Hubbard}
\end{figure}

\section{Numerical Details for Bounding Entanglement Entropies in PXP Model}
In this section, we shall elaborate on the numerical procedure we use to bound the EEs of the scar states in the PXP model.  We start to construct the correlation matrix for given range-$k$ operator basis of the form (still the identity operator is excluded)
\begin{equation}
    \hat{L}_{i} = \lambda_{a_1}^{(1)} \otimes \lambda_{a_2}^{(2)} \otimes...\otimes \lambda_{a_k}^{(k)} \, ,
\end{equation}
where $\{a_{j}\} = \{ 0, 1, 2, 3 \}$ for $j=1, 2, \cdots, k$ are the indices of the $2 \times 2$ identity matrix $\mathbb{I}$ (where $a_{j}=0$) and the three Pauli matrices $\{\sigma_{x}, \sigma_{y}, \sigma_{z}\}$. We calculate the $(4^k - 1)$-dimensional local correlation matrix for every eigenstate and compare the correlation spectrum of a scar state with surrounding thermal states. As shown in Fig.~\ref{fig:pxpl20k5}, the correlation spectra of scar states are markedly different with more smaller eigenvalues. From such plot, we can determine a position $\alpha=N_{1}$ (labeled by pink arrow in Fig.~\ref{fig:pxpl20k5}) where the eigenvalue plot of the scar state is about to cross those of the surrounding thermal states.
\begin{figure}[htbp]
\centering
\includegraphics[width=0.9\linewidth]{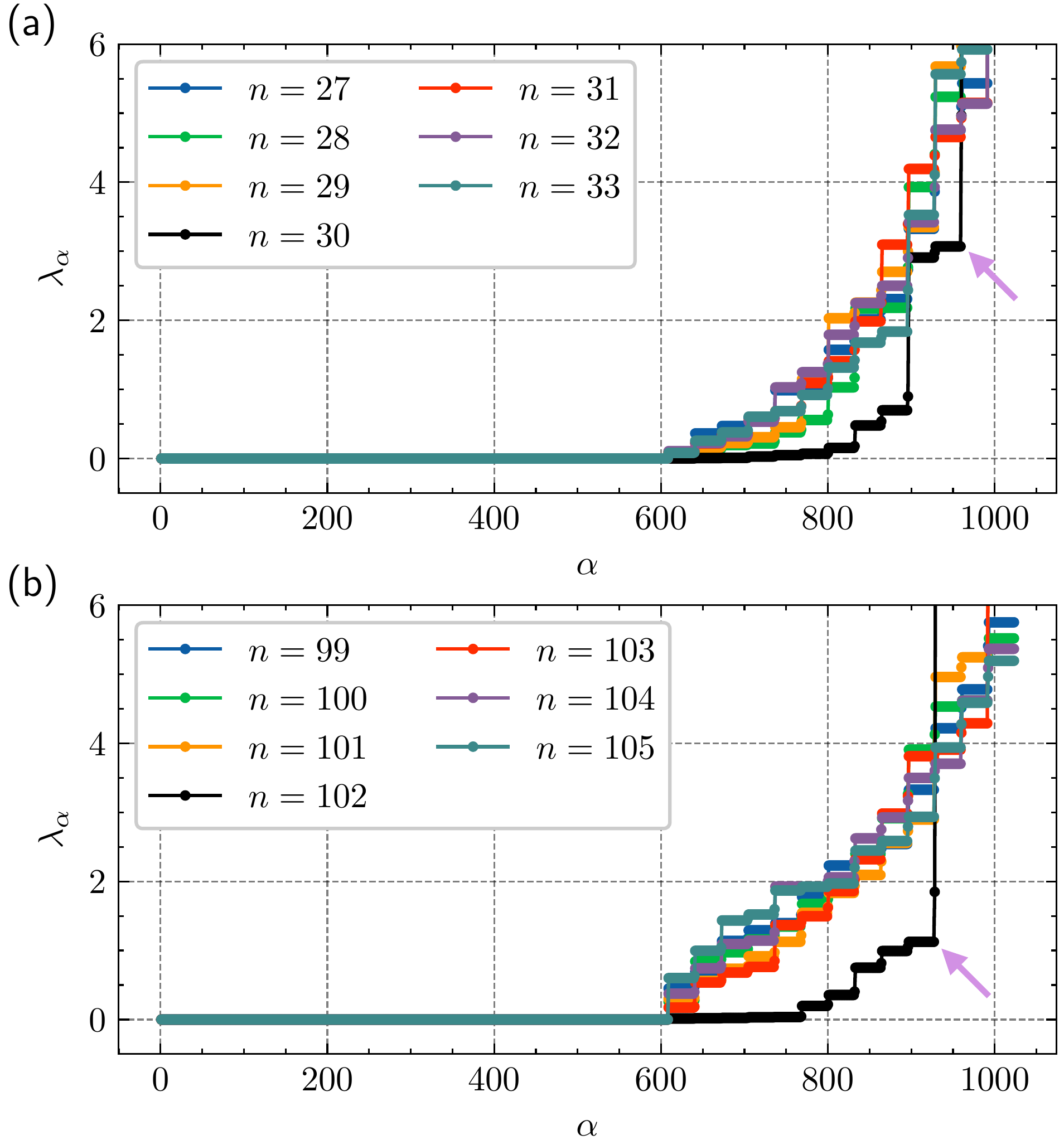}
\caption{Plots of correlation spectra for a group of eigenstates (eigenstate indices $n$ shown in the legends) around two quantum many-body scar states (marked as black) for the PXP model in the zero momentum and inversion even sector. The system size is $L=20$ and the range of the operator basis $\{\hat{L}_{i}\}$ is $k=5$. In both groups, the eigenvalues of the correlation
matrix of the scar state noticeably fall below (excluding the first $N_{0}=608$ trivial zeros) those of the surrounding thermal eigenstates when $\alpha \le N_{1}$ (corresponding points marked by pink arrows). The $N_{1} - N_{0}$ positive semidefinite $\hat{\Delta}_{\alpha}$ operators associated with these eigenvalues are then taken as the starting point to construct $\hat{H}^{A}_{\text{aux}}$ to bound EE.}
\label{fig:pxpl20k5}
\end{figure}
We shall use the $N_{1} - N_{0}$ (we subtract $N_{0}$ to exclude the trivial zeros) eigenvalues to construct an auxiliary Hamiltonian $\hat{H}_{\text{aux}}^{A}$ to bound the EE. For such nonzero eigenvalue $\lambda_{\alpha} \neq 0$, we still proceed to define an operator $\hat{O}$ using the corresponding eigenvector $\mathbf{e}_{\alpha} = (w_{1}, w_{2}, \cdots, w_{n})$ in a similar fashion, $\hat{O} = \sum_{i} w_{i} \hat{L}_{i}$. In this case, $\hat{O}$ is not an eigenoperator of $\ket{\psi}$ anymore,
\begin{equation}
    \hat{O} \ket{\psi} = \xi \ket{\psi} + \ket{\psi_{\perp}} \, , \quad \text{where } \langle\psi_{\perp} | \psi_{\perp} \rangle > 0 \, .
\end{equation}
We can, however, still define a positive semidefinite operator $\hat{\Delta}_{\alpha}$ as before,
\begin{equation}
    \hat{\Delta}_{\alpha} = \left( \hat{O}^{\dagger} - \xi^{*} \hat{I} \right) \left( \hat{O} - \xi \hat{I} \right) \, .
\end{equation}
The constructed auxiliary Hamiltonian still takes the same form
\begin{equation}
    \hat{H}_{\text{aux}}^{A} = \sum_{\alpha} \tilde{c}_{\alpha} \hat{\Delta}_{\alpha} \, , \quad  \tilde{c}_{\alpha} > 0 \, .
\end{equation}
For given $N_{1}$, optimizing the value of $\tilde{c}_{\alpha}$ to get the best upper bound is beyond the scope of our paper. Therefore, for demonstration purposes, we shall take the simplest case where $\tilde{c}_{a} = 1, \forall \alpha$. Since $N_{1}$ is handpicked, we choose to vary the value of $N_{1}$ to get a comparatively good upper bound of EE. We summarize the optimized result of the EE upper bounds of the scar states of the PXP model in Table~\ref{table:PXP} (see also  \cite{Yao:source_code}).

\begin{table*}[ht]
\begin{center}
\renewcommand{\arraystretch}{1.5}
\begin{tabular}{>{\centering}p{1.0cm}|>{\centering}p{1.2cm}>{\centering}p{0.8cm}|>{\centering}p{1.2cm}>{\centering}p{0.8cm}|>{\centering}p{1.2cm}>{\centering}p{0.8cm}|>{\centering}p{1.2cm}>{\centering}p{0.8cm}|>{\centering}p{1.2cm}>{\centering}p{0.8cm}}
   \hline \hline
   & \multicolumn{2}{c|}{$n=1$} & \multicolumn{2}{c|}{$n=2$} & \multicolumn{2}{c|}{$n=7$} & \multicolumn{2}{c|}{$n=30$}  & \multicolumn{2}{c}{$n=102$} \\ \hline
   $k = 3$ & 0.4038 & 10 & 1.6391 & 24 & 2.1314 & 12 & 2.6668 & 16 & 2.2858 & 19  \tabularnewline
   \hline
   $k = 4$ & 0.4031 & 81 & 1.5718 & 93 & 2.0201 & 76 & 2.0227 & 86 & 1.9797 & 89  \tabularnewline
   \hline
   $k = 5$ & 0.3868 &302 & 1.4835 & 330& 1.9588 &203 & 2.0090 &336 & 1.7840 & 315 \tabularnewline
   \hline \hline
\end{tabular}
\end{center}
\caption{Optimization results of entanglement entropy upper bounds of scar states (with eigenstate index $n$) of the PXP model in the zero momentum and inversion even sector. Results of the range $k=3, 4, 5$ of operator basis $\{\hat{L}_i\}$ are shown. In each cell, the left value is the optimized EE upper bound (rounded with 4 decimal points) and the right number is the corresponding optimal value of $N_{1} - N_{0}$. Due to the ``particle-hole" symmetry of the PXP model \cite{Turner:2018qs}, results from the only the first half, with eigenstate indices $n=1, 2, 7, 30, 102$, of all the scar states are shown.}
\label{table:PXP}
\end{table*}

As can be seen from either Fig.~4 in the main text or Table.~I here, the obtained optimal upper bound decrease with the increase of $k$. For $k=4, 5$, our bounds of the scar states are well separated from the thermal continuum, signifying weak ergodicity breaking \cite{Turner:2018wg}.

\section{Experimental Scheme to Measure the Correlation Matrix}
Here we will give a detailed description of the experimental measurement of correlation matrix $\mathcal{M}$ in the context of cold atom systems. We also give an explicit example of expanding the product of two operators.\\

To be concrete, we consider a spin-1/2 chain in optical lattices, which can be realized by ultracold atoms with two hyperfine states and large on-site repulsion. The system is prepared in the state $|\psi\rangle$, of which the entanglement property is to be bounded. As an example, we consider the measurement of the correlation matrix $\mathcal{M}$ of the subsystem $A$ that contains first two sites $n=1,2$. Let us choose our operator basis $\{ \hat{L}_{i} \}$ as $\{\hat{\sigma}_a^1, \hat{\sigma}_a^2, \hat{\sigma}_a^1\hat{\sigma}_b^2\}$ where $a,b=x,y,z$ labels the three Pauli operators, and add identity operator $\hat{I}$ to form the complete set of hermitian operator basis $\{ \hat{\tilde{L}}_{i} \} =  \{\hat{I}, \hat{\sigma}_a^1, \hat{\sigma}_a^2, \hat{\sigma}_a^1\hat{\sigma}_b^2\}$. To obtain the matrix element $\mathcal{M}_{ij}$, we need to measure $\langle \psi| \hat{L}_i^\dagger \hat{L}_j|\psi\rangle$ and $\langle \psi| \hat{L}_i|\psi\rangle$. However, it is straightforward to see that $\hat{L}_i^\dagger \hat{L}_j=\hat{L}_i \hat{L}_j$ is an element in $\{ \hat{\tilde{L}}_i\}$ using the fact that
\begin{equation}
    \hat{\sigma}^n_a \hat{\sigma}^n_b = \delta^{ab} \hat{I} + \sum_{c} i \epsilon_{abc} \hat{\sigma}_c^n,
\end{equation}
where $\epsilon_{abc}$ is the Levi-Civita symbol. In more general cases, when the set of operators $\{\hat L_i\}$ is orthogonal under the matrix trace $\mathrm{Tr} (\hat{L}_i \hat{L}_j)  \propto \delta_{ij}$, we can write
\begin{equation}
    \hat{L}_i \hat{L}_j=\sum_k C_{ij}^k \hat{\tilde{L}}_k, \qquad C_{ij}^k=\frac{\mathrm{Tr}(\hat{L}_i \hat{L}_j \hat{\tilde{L}}_k)}{\mathrm{Tr}(\hat{\tilde{L}}_k^2)}.
\end{equation}
As a result, knowing $\{ \langle \psi| \hat{\tilde{L}}_i|\psi\rangle \}$, or $\{ \langle \psi| \hat{L}_i|\psi\rangle \}$, is enough to determine the correlation matrix $\mathcal{M}$.
We first consider the measurement of $\hat{\sigma}^1_z$, $\hat{\sigma}^2_z$, and $\hat{\sigma}^1_z \hat{\sigma}^2_z$, which directly corresponds to the (correlation of) atom occupation in different hyperfine states. The technique of quantum gas microscopy can be directly applied to measure these quantities by fluorescence imaging  \cite{Bakr:2009aq,Sherson:2010sa,Weitenberg:2011ss,Cheuk:2015qg,Parsons:2015sr,Gross:2021qg}. Repeated measurements produce the probability distribution of spin states in the $z$ direction $p_{s_1 s_2}$ where the subscript $s_n=\uparrow,\downarrow$ labels the state of the $n$-th site. We have
\begin{equation}
\begin{aligned}
    \langle \psi|\hat{\sigma}^1_z|\psi\rangle&=\sum_{s_2}(p_{\uparrow s_2}-p_{\downarrow s_2}),\\
    \langle \psi|\hat{\sigma}^2_z|\psi\rangle&=\sum_{s_1}(p_{s_1\uparrow }-p_{s_1\downarrow }),\\
    \langle \psi|\hat{\sigma}^1_z\hat{\sigma}^2_z|\psi\rangle&=p_{\uparrow\uparrow }+p_{\downarrow\downarrow }-p_{\downarrow\uparrow }-p_{\uparrow\downarrow },
\end{aligned}
\end{equation}
which can be measured at the same time. Moreover, for systems with translation symmetry, one can use data on different sites to compute $p_{s_1 s_2}$ efficiently in a single-shot measurement.

For other operators, the measurement can be done with an additional Raman pulse that couples two hyperfine states. Such optical manipulation of a single site in optical lattices has been realized in experiments (for example, in \cite{Wang:2015ca,Xia:2015rb,Wang:2016sq}). We give a explicit example for the operator $\hat{\sigma}_x^1\hat{\sigma}_z^2$. By tuning the relative phase between Raman laser, we choose the coupling between two hyperfine states on the first site as $\delta \hat{H} = \Omega \hat{\sigma}_y^1$. We consider applying a $\pi/2$ pulse, which leads to an unitary evolution denoted by $\hat{U}$. Using the fact that $\hat{\sigma}_x^1=-\hat{U}^{\dagger}\hat{\sigma}_z^1\hat{U}$, we have
\begin{equation}
    \langle \psi |\hat{\sigma}_x^1\hat{\sigma}_z^2|\psi \rangle=-\langle \psi |\hat{U}^\dagger\hat{\sigma}_z^1\hat{\sigma}_z^2\hat{U}|\psi \rangle.
\end{equation}
As a result, after applying the $\pi/2$ pulse, the experimental protocol to measure $\hat{\sigma}_x^1\hat{\sigma}_z^2$ directly follows the protocol of measuring $\hat{\sigma}_z^1 \hat{\sigma}_z^2$. Other operators can be measured in the same spirit, and the generalization to $k$-local operators is straightforward.

\bibliography{references.bib}

\begin{thebibliography}{75}%
\makeatletter
\providecommand \@ifxundefined [1]{%
 \@ifx{#1\undefined}
}%
\providecommand \@ifnum [1]{%
 \ifnum #1\expandafter \@firstoftwo
 \else \expandafter \@secondoftwo
 \fi
}%
\providecommand \@ifx [1]{%
 \ifx #1\expandafter \@firstoftwo
 \else \expandafter \@secondoftwo
 \fi
}%
\providecommand \natexlab [1]{#1}%
\providecommand \enquote  [1]{``#1''}%
\providecommand \bibnamefont  [1]{#1}%
\providecommand \bibfnamefont [1]{#1}%
\providecommand \citenamefont [1]{#1}%
\providecommand \href@noop [0]{\@secondoftwo}%
\providecommand \href [0]{\begingroup \@sanitize@url \@href}%
\providecommand \@href[1]{\@@startlink{#1}\@@href}%
\providecommand \@@href[1]{\endgroup#1\@@endlink}%
\providecommand \@sanitize@url [0]{\catcode `\\12\catcode `\$12\catcode
  `\&12\catcode `\#12\catcode `\^12\catcode `\_12\catcode `\%12\relax}%
\providecommand \@@startlink[1]{}%
\providecommand \@@endlink[0]{}%
\providecommand \url  [0]{\begingroup\@sanitize@url \@url }%
\providecommand \@url [1]{\endgroup\@href {#1}{\urlprefix }}%
\providecommand \urlprefix  [0]{URL }%
\providecommand \Eprint [0]{\href }%
\providecommand \doibase [0]{https://doi.org/}%
\providecommand \selectlanguage [0]{\@gobble}%
\providecommand \bibinfo  [0]{\@secondoftwo}%
\providecommand \bibfield  [0]{\@secondoftwo}%
\providecommand \translation [1]{[#1]}%
\providecommand \BibitemOpen [0]{}%
\providecommand \bibitemStop [0]{}%
\providecommand \bibitemNoStop [0]{.\EOS\space}%
\providecommand \EOS [0]{\spacefactor3000\relax}%
\providecommand \BibitemShut  [1]{\csname bibitem#1\endcsname}%
\let\auto@bib@innerbib\@empty
\bibitem [{\citenamefont {Amico}\ \emph {et~al.}(2008)\citenamefont {Amico},
  \citenamefont {Fazio}, \citenamefont {Osterloh},\ and\ \citenamefont
  {Vedral}}]{Amico:2008RMP}%
  \BibitemOpen
  \bibfield  {author} {\bibinfo {author} {\bibfnamefont {L.}~\bibnamefont
  {Amico}}, \bibinfo {author} {\bibfnamefont {R.}~\bibnamefont {Fazio}},
  \bibinfo {author} {\bibfnamefont {A.}~\bibnamefont {Osterloh}},\ and\
  \bibinfo {author} {\bibfnamefont {V.}~\bibnamefont {Vedral}},\ }\bibfield
  {title} {\bibinfo {title} {{Entanglement in many-body systems}},\ }\href
  {https://doi.org/10.1103/RevModPhys.80.517} {\bibfield  {journal} {\bibinfo
  {journal} {Rev. Mod. Phys.}\ }\textbf {\bibinfo {volume} {80}},\ \bibinfo
  {pages} {517} (\bibinfo {year} {2008})}\BibitemShut {NoStop}%
\bibitem [{\citenamefont {Laflorencie}(2016)}]{Laflorencie:2016EE}%
  \BibitemOpen
  \bibfield  {author} {\bibinfo {author} {\bibfnamefont {N.}~\bibnamefont
  {Laflorencie}},\ }\bibfield  {title} {\bibinfo {title} {{Quantum entanglement
  in condensed matter systems}},\ }\href
  {https://doi.org/https://doi.org/10.1016/j.physrep.2016.06.008} {\bibfield
  {journal} {\bibinfo  {journal} {Physics Reports}\ }\textbf {\bibinfo {volume}
  {646}},\ \bibinfo {pages} {1} (\bibinfo {year} {2016})}\BibitemShut {NoStop}%
\bibitem [{\citenamefont {Eisert}\ \emph {et~al.}(2010)\citenamefont {Eisert},
  \citenamefont {Cramer},\ and\ \citenamefont {Plenio}}]{Eisert:2010ca}%
  \BibitemOpen
  \bibfield  {author} {\bibinfo {author} {\bibfnamefont {J.}~\bibnamefont
  {Eisert}}, \bibinfo {author} {\bibfnamefont {M.}~\bibnamefont {Cramer}},\
  and\ \bibinfo {author} {\bibfnamefont {M.~B.}\ \bibnamefont {Plenio}},\
  }\bibfield  {title} {\bibinfo {title} {{Colloquium: Area laws for the
  entanglement entropy}},\ }\href {https://doi.org/10.1103/RevModPhys.82.277}
  {\bibfield  {journal} {\bibinfo  {journal} {Rev. Mod. Phys.}\ }\textbf
  {\bibinfo {volume} {82}},\ \bibinfo {pages} {277} (\bibinfo {year}
  {2010})}\BibitemShut {NoStop}%
\bibitem [{\citenamefont {Abanin}\ \emph {et~al.}(2019)\citenamefont {Abanin},
  \citenamefont {Altman}, \citenamefont {Bloch},\ and\ \citenamefont
  {Serbyn}}]{Abanin:2019cm}%
  \BibitemOpen
  \bibfield  {author} {\bibinfo {author} {\bibfnamefont {D.~A.}\ \bibnamefont
  {Abanin}}, \bibinfo {author} {\bibfnamefont {E.}~\bibnamefont {Altman}},
  \bibinfo {author} {\bibfnamefont {I.}~\bibnamefont {Bloch}},\ and\ \bibinfo
  {author} {\bibfnamefont {M.}~\bibnamefont {Serbyn}},\ }\bibfield  {title}
  {\bibinfo {title} {{Colloquium: Many-body localization, thermalization, and
  entanglement}},\ }\href {https://doi.org/10.1103/RevModPhys.91.021001}
  {\bibfield  {journal} {\bibinfo  {journal} {Rev. Mod. Phys.}\ }\textbf
  {\bibinfo {volume} {91}},\ \bibinfo {pages} {021001} (\bibinfo {year}
  {2019})}\BibitemShut {NoStop}%
\bibitem [{\citenamefont {Vidal}\ \emph {et~al.}(2003)\citenamefont {Vidal},
  \citenamefont {Latorre}, \citenamefont {Rico},\ and\ \citenamefont
  {Kitaev}}]{Vidal:2003ei}%
  \BibitemOpen
  \bibfield  {author} {\bibinfo {author} {\bibfnamefont {G.}~\bibnamefont
  {Vidal}}, \bibinfo {author} {\bibfnamefont {J.~I.}\ \bibnamefont {Latorre}},
  \bibinfo {author} {\bibfnamefont {E.}~\bibnamefont {Rico}},\ and\ \bibinfo
  {author} {\bibfnamefont {A.~Y.}\ \bibnamefont {Kitaev}},\ }\bibfield  {title}
  {\bibinfo {title} {{Entanglement in Quantum Critical Phenomena}},\ }\href
  {https://doi.org/10.1103/PhysRevLett.90.227902} {\bibfield  {journal}
  {\bibinfo  {journal} {Phys. Rev. Lett.}\ }\textbf {\bibinfo {volume} {90}},\
  \bibinfo {pages} {227902} (\bibinfo {year} {2003})}\BibitemShut {NoStop}%
\bibitem [{\citenamefont {Sachdev}(2008)}]{Sachdev:2008qm}%
  \BibitemOpen
  \bibfield  {author} {\bibinfo {author} {\bibfnamefont {S.}~\bibnamefont
  {Sachdev}},\ }\bibfield  {title} {\bibinfo {title} {{Quantum magnetism and
  criticality}},\ }\href {https://doi.org/10.1038/nphys894} {\bibfield
  {journal} {\bibinfo  {journal} {Nat. Phys.}\ }\textbf {\bibinfo {volume}
  {4}},\ \bibinfo {pages} {173} (\bibinfo {year} {2008})}\BibitemShut {NoStop}%
\bibitem [{\citenamefont {Holzhey}\ \emph {et~al.}(1994)\citenamefont
  {Holzhey}, \citenamefont {Larsen},\ and\ \citenamefont
  {Wilczek}}]{HOLZHEY1994443}%
  \BibitemOpen
  \bibfield  {author} {\bibinfo {author} {\bibfnamefont {C.}~\bibnamefont
  {Holzhey}}, \bibinfo {author} {\bibfnamefont {F.}~\bibnamefont {Larsen}},\
  and\ \bibinfo {author} {\bibfnamefont {F.}~\bibnamefont {Wilczek}},\
  }\bibfield  {title} {\bibinfo {title} {Geometric and renormalized entropy in
  conformal field theory},\ }\href
  {https://doi.org/https://doi.org/10.1016/0550-3213(94)90402-2} {\bibfield
  {journal} {\bibinfo  {journal} {Nuclear Physics B}\ }\textbf {\bibinfo
  {volume} {424}},\ \bibinfo {pages} {443} (\bibinfo {year}
  {1994})}\BibitemShut {NoStop}%
\bibitem [{\citenamefont {Calabrese}\ and\ \citenamefont
  {Cardy}(2004)}]{Calabrese2004}%
  \BibitemOpen
  \bibfield  {author} {\bibinfo {author} {\bibfnamefont {P.}~\bibnamefont
  {Calabrese}}\ and\ \bibinfo {author} {\bibfnamefont {J.}~\bibnamefont
  {Cardy}},\ }\bibfield  {title} {\bibinfo {title} {Entanglement entropy and
  quantum field theory},\ }\href
  {http://stacks.iop.org/1742-5468/2004/i=06/a=P06002} {\bibfield  {journal}
  {\bibinfo  {journal} {J. Stat. Mech.}\ }\textbf {\bibinfo {volume} {2004}},\
  \bibinfo {pages} {P06002} (\bibinfo {year} {2004})}\BibitemShut {NoStop}%
\bibitem [{\citenamefont {Calabrese}\ and\ \citenamefont
  {Cardy}(2009)}]{Calabrese2009}%
  \BibitemOpen
  \bibfield  {author} {\bibinfo {author} {\bibfnamefont {P.}~\bibnamefont
  {Calabrese}}\ and\ \bibinfo {author} {\bibfnamefont {J.}~\bibnamefont
  {Cardy}},\ }\bibfield  {title} {\bibinfo {title} {Entanglement entropy and
  conformal field theory},\ }\href
  {http://stacks.iop.org/1751-8121/42/i=50/a=504005} {\bibfield  {journal}
  {\bibinfo  {journal} {J. Phys. A: Math. Theor.}\ }\textbf {\bibinfo {volume}
  {42}},\ \bibinfo {pages} {504005} (\bibinfo {year} {2009})}\BibitemShut
  {NoStop}%
\bibitem [{\citenamefont {Kitaev}\ and\ \citenamefont
  {Preskill}(2006)}]{Kitaev:2006PRL_EE}%
  \BibitemOpen
  \bibfield  {author} {\bibinfo {author} {\bibfnamefont {A.~Y.}\ \bibnamefont
  {Kitaev}}\ and\ \bibinfo {author} {\bibfnamefont {J.}~\bibnamefont
  {Preskill}},\ }\bibfield  {title} {\bibinfo {title} {{Topological
  Entanglement Entropy}},\ }\href
  {https://doi.org/10.1103/PhysRevLett.96.110404} {\bibfield  {journal}
  {\bibinfo  {journal} {Phys. Rev. Lett.}\ }\textbf {\bibinfo {volume} {96}},\
  \bibinfo {pages} {110404} (\bibinfo {year} {2006})}\BibitemShut {NoStop}%
\bibitem [{\citenamefont {Levin}\ and\ \citenamefont
  {Wen}(2006)}]{Levin:2006PRL_EE}%
  \BibitemOpen
  \bibfield  {author} {\bibinfo {author} {\bibfnamefont {M.~A.}\ \bibnamefont
  {Levin}}\ and\ \bibinfo {author} {\bibfnamefont {X.-G.}\ \bibnamefont
  {Wen}},\ }\bibfield  {title} {\bibinfo {title} {{Detecting Topological Order
  in a Ground State Wave Function}},\ }\href
  {https://doi.org/10.1103/PhysRevLett.96.110405} {\bibfield  {journal}
  {\bibinfo  {journal} {Phys. Rev. Lett.}\ }\textbf {\bibinfo {volume} {96}},\
  \bibinfo {pages} {110405} (\bibinfo {year} {2006})}\BibitemShut {NoStop}%
\bibitem [{\citenamefont {Nishioka}\ \emph {et~al.}(2009)\citenamefont
  {Nishioka}, \citenamefont {Ryu},\ and\ \citenamefont
  {Takayanagi}}]{Nishioka:2009he}%
  \BibitemOpen
  \bibfield  {author} {\bibinfo {author} {\bibfnamefont {T.}~\bibnamefont
  {Nishioka}}, \bibinfo {author} {\bibfnamefont {S.}~\bibnamefont {Ryu}},\ and\
  \bibinfo {author} {\bibfnamefont {T.}~\bibnamefont {Takayanagi}},\ }\bibfield
   {title} {\bibinfo {title} {{Holographic entanglement entropy: an
  overview}},\ }\href {https://doi.org/10.1088/1751-8113/42/50/504008}
  {\bibfield  {journal} {\bibinfo  {journal} {J. Phys. A: Math. Theor.}\
  }\textbf {\bibinfo {volume} {42}},\ \bibinfo {pages} {504008} (\bibinfo
  {year} {2009})}\BibitemShut {NoStop}%
\bibitem [{\citenamefont {Zaanen}\ \emph {et~al.}(2015)\citenamefont {Zaanen},
  \citenamefont {Liu}, \citenamefont {Sun},\ and\ \citenamefont
  {Schalm}}]{Zaanen:Book}%
  \BibitemOpen
  \bibfield  {author} {\bibinfo {author} {\bibfnamefont {J.}~\bibnamefont
  {Zaanen}}, \bibinfo {author} {\bibfnamefont {Y.}~\bibnamefont {Liu}},
  \bibinfo {author} {\bibfnamefont {Y.-W.}\ \bibnamefont {Sun}},\ and\ \bibinfo
  {author} {\bibfnamefont {K.}~\bibnamefont {Schalm}},\ }\href
  {https://doi.org/10.1017/CBO9781139942492} {\emph {\bibinfo {title}
  {{Holographic Duality in Condensed Matter Physics}}}}\ (\bibinfo  {publisher}
  {Cambridge University Press},\ \bibinfo {address} {Cambridge},\ \bibinfo
  {year} {2015})\BibitemShut {NoStop}%
\bibitem [{\citenamefont {Hartnoll}\ \emph {et~al.}(2018)\citenamefont
  {Hartnoll}, \citenamefont {Lucas},\ and\ \citenamefont
  {Sachdev}}]{Hartnoll:Book}%
  \BibitemOpen
  \bibfield  {author} {\bibinfo {author} {\bibfnamefont {S.~A.}\ \bibnamefont
  {Hartnoll}}, \bibinfo {author} {\bibfnamefont {A.}~\bibnamefont {Lucas}},\
  and\ \bibinfo {author} {\bibfnamefont {S.}~\bibnamefont {Sachdev}},\ }\href
  {https://mitpress.mit.edu/books/holographic-quantum-matter} {\emph {\bibinfo
  {title} {{Holographic Quantum Matter}}}}\ (\bibinfo  {publisher} {MIT
  Press},\ \bibinfo {address} {Cambridge, MA},\ \bibinfo {year}
  {2018})\BibitemShut {NoStop}%
\bibitem [{\citenamefont {Chaikin}\ and\ \citenamefont
  {Lubensky}(1995)}]{Chaikin:2000Book}%
  \BibitemOpen
  \bibfield  {author} {\bibinfo {author} {\bibfnamefont {P.~M.}\ \bibnamefont
  {Chaikin}}\ and\ \bibinfo {author} {\bibfnamefont {T.~C.}\ \bibnamefont
  {Lubensky}},\ }\href {https://doi.org/10.1017/CBO9780511813467} {\emph
  {\bibinfo {title} {{Principles of Condensed Matter Physics}}}}\ (\bibinfo
  {publisher} {Cambridge University Press},\ \bibinfo {address} {Cambridge},\
  \bibinfo {year} {1995})\BibitemShut {NoStop}%
\bibitem [{\citenamefont {Wolf}\ \emph {et~al.}(2008)\citenamefont {Wolf},
  \citenamefont {Verstraete}, \citenamefont {Hastings},\ and\ \citenamefont
  {Cirac}}]{Wolf:2008PRL}%
  \BibitemOpen
  \bibfield  {author} {\bibinfo {author} {\bibfnamefont {M.~M.}\ \bibnamefont
  {Wolf}}, \bibinfo {author} {\bibfnamefont {F.}~\bibnamefont {Verstraete}},
  \bibinfo {author} {\bibfnamefont {M.~B.}\ \bibnamefont {Hastings}},\ and\
  \bibinfo {author} {\bibfnamefont {J.~I.}\ \bibnamefont {Cirac}},\ }\bibfield
  {title} {\bibinfo {title} {{Area Laws in Quantum Systems: Mutual Information
  and Correlations}},\ }\href {https://doi.org/10.1103/PhysRevLett.100.070502}
  {\bibfield  {journal} {\bibinfo  {journal} {Phys. Rev. Lett.}\ }\textbf
  {\bibinfo {volume} {100}},\ \bibinfo {pages} {070502} (\bibinfo {year}
  {2008})}\BibitemShut {NoStop}%
\bibitem [{\citenamefont {Qi}(2021)}]{Xiaoliang_talk}%
  \BibitemOpen
  \bibfield  {author} {\bibinfo {author} {\bibfnamefont {X.}~\bibnamefont
  {Qi}},\ }\bibfield  {title} {\bibinfo {title} {Quantum information measure of
  space-time correlation},\ }\href@noop {} {\bibfield  {journal} {\bibinfo
  {journal} {talk given at Tsinghua Univeristy}\ } (\bibinfo {year}
  {2021})}\BibitemShut {NoStop}%
\bibitem [{\citenamefont {Fan}\ \emph {et~al.}(2017)\citenamefont {Fan},
  \citenamefont {Zhang}, \citenamefont {Shen},\ and\ \citenamefont
  {Zhai}}]{Fan:2017OTOC}%
  \BibitemOpen
  \bibfield  {author} {\bibinfo {author} {\bibfnamefont {R.}~\bibnamefont
  {Fan}}, \bibinfo {author} {\bibfnamefont {P.}~\bibnamefont {Zhang}}, \bibinfo
  {author} {\bibfnamefont {H.}~\bibnamefont {Shen}},\ and\ \bibinfo {author}
  {\bibfnamefont {H.}~\bibnamefont {Zhai}},\ }\bibfield  {title} {\bibinfo
  {title} {{Out-of-time-order correlation for many-body localization}},\ }\href
  {https://doi.org/10.1016/j.scib.2017.04.011} {\bibfield  {journal} {\bibinfo
  {journal} {Science Bulletin}\ }\textbf {\bibinfo {volume} {62}},\ \bibinfo
  {pages} {707} (\bibinfo {year} {2017})}\BibitemShut {NoStop}%
\bibitem [{\citenamefont {Swingle}\ \emph {et~al.}(2016)\citenamefont
  {Swingle}, \citenamefont {Bentsen}, \citenamefont {Schleier-Smith},\ and\
  \citenamefont {Hayden}}]{Swingle:2016PRA}%
  \BibitemOpen
  \bibfield  {author} {\bibinfo {author} {\bibfnamefont {B.}~\bibnamefont
  {Swingle}}, \bibinfo {author} {\bibfnamefont {G.}~\bibnamefont {Bentsen}},
  \bibinfo {author} {\bibfnamefont {M.}~\bibnamefont {Schleier-Smith}},\ and\
  \bibinfo {author} {\bibfnamefont {P.}~\bibnamefont {Hayden}},\ }\bibfield
  {title} {\bibinfo {title} {{Measuring the scrambling of quantum
  information}},\ }\href {https://doi.org/10.1103/PhysRevA.94.040302}
  {\bibfield  {journal} {\bibinfo  {journal} {Phys. Rev. A}\ }\textbf {\bibinfo
  {volume} {94}},\ \bibinfo {pages} {040302} (\bibinfo {year}
  {2016})}\BibitemShut {NoStop}%
\bibitem [{\citenamefont {Chen}(2016)}]{Chen:2016ul}%
  \BibitemOpen
  \bibfield  {author} {\bibinfo {author} {\bibfnamefont {Y.}~\bibnamefont
  {Chen}},\ }\href@noop {} {\bibinfo {title} {Universal logarithmic scrambling
  in many body localization}} (\bibinfo {year} {2016}),\ \Eprint
  {https://arxiv.org/abs/1608.02765} {arXiv:1608.02765 [cond-mat.dis-nn]}
  \BibitemShut {NoStop}%
\bibitem [{\citenamefont {Chen}\ \emph {et~al.}(2017)\citenamefont {Chen},
  \citenamefont {Zhou}, \citenamefont {Huse},\ and\ \citenamefont
  {Fradkin}}]{Chen:2017OTOC}%
  \BibitemOpen
  \bibfield  {author} {\bibinfo {author} {\bibfnamefont {X.}~\bibnamefont
  {Chen}}, \bibinfo {author} {\bibfnamefont {T.}~\bibnamefont {Zhou}}, \bibinfo
  {author} {\bibfnamefont {D.~A.}\ \bibnamefont {Huse}},\ and\ \bibinfo
  {author} {\bibfnamefont {E.}~\bibnamefont {Fradkin}},\ }\bibfield  {title}
  {\bibinfo {title} {{Out-of-time-order correlations in many-body localized and
  thermal phases}},\ }\href {https://doi.org/10.1002/andp.201600332} {\bibfield
   {journal} {\bibinfo  {journal} {Annalen Der Physik}\ }\textbf {\bibinfo
  {volume} {529}},\ \bibinfo {pages} {1600332} (\bibinfo {year}
  {2017})}\BibitemShut {NoStop}%
\bibitem [{\citenamefont {Huang}\ \emph {et~al.}(2017)\citenamefont {Huang},
  \citenamefont {Zhang},\ and\ \citenamefont {Chen}}]{Huang:2017OTOC}%
  \BibitemOpen
  \bibfield  {author} {\bibinfo {author} {\bibfnamefont {Y.}~\bibnamefont
  {Huang}}, \bibinfo {author} {\bibfnamefont {Y.~L.}\ \bibnamefont {Zhang}},\
  and\ \bibinfo {author} {\bibfnamefont {X.}~\bibnamefont {Chen}},\ }\bibfield
  {title} {\bibinfo {title} {{Out-of-time-ordered correlators in many-body
  localized systems}},\ }\href {https://doi.org/10.1002/andp.201600318}
  {\bibfield  {journal} {\bibinfo  {journal} {Annalen Der Physik}\ }\textbf
  {\bibinfo {volume} {529}},\ \bibinfo {pages} {1600318} (\bibinfo {year}
  {2017})}\BibitemShut {NoStop}%
\bibitem [{\citenamefont {He}\ and\ \citenamefont {Lu}(2017)}]{He:2017PRB}%
  \BibitemOpen
  \bibfield  {author} {\bibinfo {author} {\bibfnamefont {R.-Q.}\ \bibnamefont
  {He}}\ and\ \bibinfo {author} {\bibfnamefont {Z.-Y.}\ \bibnamefont {Lu}},\
  }\bibfield  {title} {\bibinfo {title} {{Characterizing many-body localization
  by out-of-time-ordered correlation}},\ }\href
  {https://doi.org/10.1103/PhysRevB.95.054201} {\bibfield  {journal} {\bibinfo
  {journal} {Phys. Rev. B}\ }\textbf {\bibinfo {volume} {95}},\ \bibinfo
  {pages} {054201} (\bibinfo {year} {2017})}\BibitemShut {NoStop}%
\bibitem [{\citenamefont {Ekert}\ \emph {et~al.}(2002)\citenamefont {Ekert},
  \citenamefont {Alves}, \citenamefont {Oi}, \citenamefont {Horodecki},
  \citenamefont {Horodecki},\ and\ \citenamefont {Kwek}}]{Ekert:2002de}%
  \BibitemOpen
  \bibfield  {author} {\bibinfo {author} {\bibfnamefont {A.~K.}\ \bibnamefont
  {Ekert}}, \bibinfo {author} {\bibfnamefont {C.~M.}\ \bibnamefont {Alves}},
  \bibinfo {author} {\bibfnamefont {D.~K.~L.}\ \bibnamefont {Oi}}, \bibinfo
  {author} {\bibfnamefont {M.}~\bibnamefont {Horodecki}}, \bibinfo {author}
  {\bibfnamefont {P.}~\bibnamefont {Horodecki}},\ and\ \bibinfo {author}
  {\bibfnamefont {L.~C.}\ \bibnamefont {Kwek}},\ }\bibfield  {title} {\bibinfo
  {title} {{Direct Estimations of Linear and Nonlinear Functionals of a Quantum
  State}},\ }\href {https://doi.org/10.1103/PhysRevLett.88.217901} {\bibfield
  {journal} {\bibinfo  {journal} {Phys. Rev. Lett.}\ }\textbf {\bibinfo
  {volume} {88}},\ \bibinfo {pages} {217901} (\bibinfo {year}
  {2002})}\BibitemShut {NoStop}%
\bibitem [{\citenamefont {Alves}\ and\ \citenamefont
  {Jaksch}(2004)}]{Alves:2004me}%
  \BibitemOpen
  \bibfield  {author} {\bibinfo {author} {\bibfnamefont {C.~M.}\ \bibnamefont
  {Alves}}\ and\ \bibinfo {author} {\bibfnamefont {D.}~\bibnamefont {Jaksch}},\
  }\bibfield  {title} {\bibinfo {title} {{Multipartite Entanglement Detection
  in Bosons}},\ }\href {https://doi.org/10.1103/PhysRevLett.93.110501}
  {\bibfield  {journal} {\bibinfo  {journal} {Phys. Rev. Lett.}\ }\textbf
  {\bibinfo {volume} {93}},\ \bibinfo {pages} {110501} (\bibinfo {year}
  {2004})}\BibitemShut {NoStop}%
\bibitem [{\citenamefont {Daley}\ \emph {et~al.}(2012)\citenamefont {Daley},
  \citenamefont {Pichler}, \citenamefont {Schachenmayer},\ and\ \citenamefont
  {Zoller}}]{Daley:2012me}%
  \BibitemOpen
  \bibfield  {author} {\bibinfo {author} {\bibfnamefont {A.~J.}\ \bibnamefont
  {Daley}}, \bibinfo {author} {\bibfnamefont {H.}~\bibnamefont {Pichler}},
  \bibinfo {author} {\bibfnamefont {J.}~\bibnamefont {Schachenmayer}},\ and\
  \bibinfo {author} {\bibfnamefont {P.}~\bibnamefont {Zoller}},\ }\bibfield
  {title} {\bibinfo {title} {{Measuring Entanglement Growth in Quench Dynamics
  of Bosons in an Optical Lattice}},\ }\href
  {https://doi.org/10.1103/PhysRevLett.109.020505} {\bibfield  {journal}
  {\bibinfo  {journal} {Phys. Rev. Lett.}\ }\textbf {\bibinfo {volume} {109}},\
  \bibinfo {pages} {020505} (\bibinfo {year} {2012})}\BibitemShut {NoStop}%
\bibitem [{\citenamefont {Islam}\ \emph {et~al.}(2015)\citenamefont {Islam},
  \citenamefont {Ma}, \citenamefont {Preiss}, \citenamefont {Tai},
  \citenamefont {Lukin}, \citenamefont {Rispoli},\ and\ \citenamefont
  {Greiner}}]{Islam:2015cm}%
  \BibitemOpen
  \bibfield  {author} {\bibinfo {author} {\bibfnamefont {R.}~\bibnamefont
  {Islam}}, \bibinfo {author} {\bibfnamefont {R.}~\bibnamefont {Ma}}, \bibinfo
  {author} {\bibfnamefont {P.~M.}\ \bibnamefont {Preiss}}, \bibinfo {author}
  {\bibfnamefont {M.~E.}\ \bibnamefont {Tai}}, \bibinfo {author} {\bibfnamefont
  {A.}~\bibnamefont {Lukin}}, \bibinfo {author} {\bibfnamefont
  {M.}~\bibnamefont {Rispoli}},\ and\ \bibinfo {author} {\bibfnamefont
  {M.}~\bibnamefont {Greiner}},\ }\bibfield  {title} {\bibinfo {title}
  {{Measuring entanglement entropy in a quantum many-body system}},\ }\href
  {https://doi.org/10.1038/nature15750} {\bibfield  {journal} {\bibinfo
  {journal} {Nature}\ }\textbf {\bibinfo {volume} {528}},\ \bibinfo {pages}
  {77} (\bibinfo {year} {2015})}\BibitemShut {NoStop}%
\bibitem [{\citenamefont {Kaufman}\ \emph {et~al.}(2016)\citenamefont
  {Kaufman}, \citenamefont {Tai}, \citenamefont {Lukin}, \citenamefont
  {Rispoli}, \citenamefont {Schittko}, \citenamefont {Preiss},\ and\
  \citenamefont {Greiner}}]{Kaufman:2016qt}%
  \BibitemOpen
  \bibfield  {author} {\bibinfo {author} {\bibfnamefont {A.~M.}\ \bibnamefont
  {Kaufman}}, \bibinfo {author} {\bibfnamefont {M.~E.}\ \bibnamefont {Tai}},
  \bibinfo {author} {\bibfnamefont {A.}~\bibnamefont {Lukin}}, \bibinfo
  {author} {\bibfnamefont {M.}~\bibnamefont {Rispoli}}, \bibinfo {author}
  {\bibfnamefont {R.}~\bibnamefont {Schittko}}, \bibinfo {author}
  {\bibfnamefont {P.~M.}\ \bibnamefont {Preiss}},\ and\ \bibinfo {author}
  {\bibfnamefont {M.}~\bibnamefont {Greiner}},\ }\bibfield  {title} {\bibinfo
  {title} {{Quantum thermalization through entanglement in an isolated
  many-body system}},\ }\href {https://doi.org/10.1126/science.aaf6725}
  {\bibfield  {journal} {\bibinfo  {journal} {Science}\ }\textbf {\bibinfo
  {volume} {353}},\ \bibinfo {pages} {794} (\bibinfo {year}
  {2016})}\BibitemShut {NoStop}%
\bibitem [{\citenamefont {van Enk}\ and\ \citenamefont
  {Beenakker}(2012)}]{Enk:2012hr}%
  \BibitemOpen
  \bibfield  {author} {\bibinfo {author} {\bibfnamefont {S.~J.}\ \bibnamefont
  {van Enk}}\ and\ \bibinfo {author} {\bibfnamefont {C.~W.~J.}\ \bibnamefont
  {Beenakker}},\ }\bibfield  {title} {\bibinfo {title} {{Measuring $\mathrm{Tr}
  \rho^n$ on Single Copies of $\rho$ Using Random Measurements}},\ }\href
  {https://doi.org/10.1103/PhysRevLett.108.110503} {\bibfield  {journal}
  {\bibinfo  {journal} {Phys. Rev. Lett.}\ }\textbf {\bibinfo {volume} {108}},\
  \bibinfo {pages} {110503} (\bibinfo {year} {2012})}\BibitemShut {NoStop}%
\bibitem [{\citenamefont {Elben}\ \emph {et~al.}(2018)\citenamefont {Elben},
  \citenamefont {Vermersch}, \citenamefont {Dalmonte}, \citenamefont {Cirac},\
  and\ \citenamefont {Zoller}}]{Elben:2018re}%
  \BibitemOpen
  \bibfield  {author} {\bibinfo {author} {\bibfnamefont {A.}~\bibnamefont
  {Elben}}, \bibinfo {author} {\bibfnamefont {B.}~\bibnamefont {Vermersch}},
  \bibinfo {author} {\bibfnamefont {M.}~\bibnamefont {Dalmonte}}, \bibinfo
  {author} {\bibfnamefont {J.~I.}\ \bibnamefont {Cirac}},\ and\ \bibinfo
  {author} {\bibfnamefont {P.}~\bibnamefont {Zoller}},\ }\bibfield  {title}
  {\bibinfo {title} {{R\'enyi Entropies from Random Quenches in Atomic Hubbard
  and Spin Models}},\ }\href {https://doi.org/10.1103/PhysRevLett.120.050406}
  {\bibfield  {journal} {\bibinfo  {journal} {Phys. Rev. Lett.}\ }\textbf
  {\bibinfo {volume} {120}},\ \bibinfo {pages} {050406} (\bibinfo {year}
  {2018})}\BibitemShut {NoStop}%
\bibitem [{\citenamefont {Vermersch}\ \emph {et~al.}(2018)\citenamefont
  {Vermersch}, \citenamefont {Elben}, \citenamefont {Dalmonte}, \citenamefont
  {Cirac},\ and\ \citenamefont {Zoller}}]{Vermersch:2018un}%
  \BibitemOpen
  \bibfield  {author} {\bibinfo {author} {\bibfnamefont {B.}~\bibnamefont
  {Vermersch}}, \bibinfo {author} {\bibfnamefont {A.}~\bibnamefont {Elben}},
  \bibinfo {author} {\bibfnamefont {M.}~\bibnamefont {Dalmonte}}, \bibinfo
  {author} {\bibfnamefont {J.~I.}\ \bibnamefont {Cirac}},\ and\ \bibinfo
  {author} {\bibfnamefont {P.}~\bibnamefont {Zoller}},\ }\bibfield  {title}
  {\bibinfo {title} {{Unitary $n$-designs via random quenches in atomic Hubbard
  and spin models: Application to the measurement of R\'enyi entropies}},\
  }\href {https://doi.org/10.1103/PhysRevA.97.023604} {\bibfield  {journal}
  {\bibinfo  {journal} {Phys. Rev. A}\ }\textbf {\bibinfo {volume} {97}},\
  \bibinfo {pages} {023604} (\bibinfo {year} {2018})}\BibitemShut {NoStop}%
\bibitem [{\citenamefont {Elben}\ \emph {et~al.}(2019)\citenamefont {Elben},
  \citenamefont {Vermersch}, \citenamefont {Roos},\ and\ \citenamefont
  {Zoller}}]{Elben:2019sc}%
  \BibitemOpen
  \bibfield  {author} {\bibinfo {author} {\bibfnamefont {A.}~\bibnamefont
  {Elben}}, \bibinfo {author} {\bibfnamefont {B.}~\bibnamefont {Vermersch}},
  \bibinfo {author} {\bibfnamefont {C.~F.}\ \bibnamefont {Roos}},\ and\
  \bibinfo {author} {\bibfnamefont {P.}~\bibnamefont {Zoller}},\ }\bibfield
  {title} {\bibinfo {title} {Statistical correlations between locally
  randomized measurements: A toolbox for probing entanglement in many-body
  quantum states},\ }\href {https://doi.org/10.1103/PhysRevA.99.052323}
  {\bibfield  {journal} {\bibinfo  {journal} {Phys. Rev. A}\ }\textbf {\bibinfo
  {volume} {99}},\ \bibinfo {pages} {052323} (\bibinfo {year}
  {2019})}\BibitemShut {NoStop}%
\bibitem [{\citenamefont {Brydges}\ \emph {et~al.}(2019)\citenamefont
  {Brydges}, \citenamefont {Elben}, \citenamefont {Jurcevic}, \citenamefont
  {Vermersch}, \citenamefont {Maier}, \citenamefont {Lanyon}, \citenamefont
  {Zoller}, \citenamefont {Blatt},\ and\ \citenamefont
  {Roos}}]{Brydges:2020pr}%
  \BibitemOpen
  \bibfield  {author} {\bibinfo {author} {\bibfnamefont {T.}~\bibnamefont
  {Brydges}}, \bibinfo {author} {\bibfnamefont {A.}~\bibnamefont {Elben}},
  \bibinfo {author} {\bibfnamefont {P.}~\bibnamefont {Jurcevic}}, \bibinfo
  {author} {\bibfnamefont {B.}~\bibnamefont {Vermersch}}, \bibinfo {author}
  {\bibfnamefont {C.}~\bibnamefont {Maier}}, \bibinfo {author} {\bibfnamefont
  {B.~P.}\ \bibnamefont {Lanyon}}, \bibinfo {author} {\bibfnamefont
  {P.}~\bibnamefont {Zoller}}, \bibinfo {author} {\bibfnamefont
  {R.}~\bibnamefont {Blatt}},\ and\ \bibinfo {author} {\bibfnamefont {C.~F.}\
  \bibnamefont {Roos}},\ }\bibfield  {title} {\bibinfo {title} {{Probing
  R\'{e}nyi entanglement entropy via randomized measurements}},\ }\href
  {https://doi.org/10.1126/science.aau4963} {\bibfield  {journal} {\bibinfo
  {journal} {Science}\ }\textbf {\bibinfo {volume} {364}},\ \bibinfo {pages}
  {260} (\bibinfo {year} {2019})}\BibitemShut {NoStop}%
\bibitem [{\citenamefont {Satzinger}\ \emph {et~al.}(2021)\citenamefont
  {Satzinger}, \citenamefont {Liu}, \citenamefont {Smith}, \citenamefont
  {Knapp}, \citenamefont {Newman}, \citenamefont {Jones}, \citenamefont {Chen},
  \citenamefont {Quintana}, \citenamefont {Mi}, \citenamefont {Dunsworth},
  \citenamefont {Gidney}, \citenamefont {Aleiner}, \citenamefont {Arute},
  \citenamefont {Arya}, \citenamefont {Atalaya}, \citenamefont {Babbush},
  \citenamefont {Bardin}, \citenamefont {Barends}, \citenamefont {Basso},
  \citenamefont {Bengtsson}, \citenamefont {Bilmes}, \citenamefont {Broughton},
  \citenamefont {Buckley}, \citenamefont {Buell}, \citenamefont {Burkett},
  \citenamefont {Bushnell}, \citenamefont {Chiaro}, \citenamefont {Collins},
  \citenamefont {Courtney}, \citenamefont {Demura}, \citenamefont {Derk},
  \citenamefont {Eppens}, \citenamefont {Erickson}, \citenamefont {Faoro},
  \citenamefont {Farhi}, \citenamefont {Fowler}, \citenamefont {Foxen},
  \citenamefont {Giustina}, \citenamefont {Greene}, \citenamefont {Gross},
  \citenamefont {Harrigan}, \citenamefont {Harrington}, \citenamefont {Hilton},
  \citenamefont {Hong}, \citenamefont {Huang}, \citenamefont {Huggins},
  \citenamefont {Ioffe}, \citenamefont {Isakov}, \citenamefont {Jeffrey},
  \citenamefont {Jiang}, \citenamefont {Kafri}, \citenamefont {Kechedzhi},
  \citenamefont {Khattar}, \citenamefont {Kim}, \citenamefont {Klimov},
  \citenamefont {Korotkov}, \citenamefont {Kostritsa}, \citenamefont
  {Landhuis}, \citenamefont {Laptev}, \citenamefont {Locharla}, \citenamefont
  {Lucero}, \citenamefont {Martin}, \citenamefont {McClean}, \citenamefont
  {McEwen}, \citenamefont {Miao}, \citenamefont {Mohseni}, \citenamefont
  {Montazeri}, \citenamefont {Mruczkiewicz}, \citenamefont {Mutus},
  \citenamefont {Naaman}, \citenamefont {Neeley}, \citenamefont {Neill},
  \citenamefont {Niu}, \citenamefont {O'Brien}, \citenamefont {Opremcak},
  \citenamefont {Pat{\'o}}, \citenamefont {Petukhov}, \citenamefont {Rubin},
  \citenamefont {Sank}, \citenamefont {Shvarts}, \citenamefont {Strain},
  \citenamefont {Szalay}, \citenamefont {Villalonga}, \citenamefont {White},
  \citenamefont {Yao}, \citenamefont {Yeh}, \citenamefont {Yoo}, \citenamefont
  {Zalcman}, \citenamefont {Neven}, \citenamefont {Boixo}, \citenamefont
  {Megrant}, \citenamefont {Chen}, \citenamefont {Kelly}, \citenamefont
  {Smelyanskiy}, \citenamefont {Kitaev}, \citenamefont {Knap}, \citenamefont
  {Pollmann},\ and\ \citenamefont {Roushan}}]{Satzinger:2021Science}%
  \BibitemOpen
  \bibfield  {author} {\bibinfo {author} {\bibfnamefont {K.~J.}\ \bibnamefont
  {Satzinger}}, \bibinfo {author} {\bibfnamefont {Y.~J.}\ \bibnamefont {Liu}},
  \bibinfo {author} {\bibfnamefont {A.}~\bibnamefont {Smith}}, \bibinfo
  {author} {\bibfnamefont {C.}~\bibnamefont {Knapp}}, \bibinfo {author}
  {\bibfnamefont {M.}~\bibnamefont {Newman}}, \bibinfo {author} {\bibfnamefont
  {C.}~\bibnamefont {Jones}}, \bibinfo {author} {\bibfnamefont
  {Z.}~\bibnamefont {Chen}}, \bibinfo {author} {\bibfnamefont {C.}~\bibnamefont
  {Quintana}}, \bibinfo {author} {\bibfnamefont {X.}~\bibnamefont {Mi}},
  \bibinfo {author} {\bibfnamefont {A.}~\bibnamefont {Dunsworth}}, \bibinfo
  {author} {\bibfnamefont {C.}~\bibnamefont {Gidney}}, \bibinfo {author}
  {\bibfnamefont {I.}~\bibnamefont {Aleiner}}, \bibinfo {author} {\bibfnamefont
  {F.}~\bibnamefont {Arute}}, \bibinfo {author} {\bibfnamefont
  {K.}~\bibnamefont {Arya}}, \bibinfo {author} {\bibfnamefont {J.}~\bibnamefont
  {Atalaya}}, \bibinfo {author} {\bibfnamefont {R.}~\bibnamefont {Babbush}},
  \bibinfo {author} {\bibfnamefont {J.~C.}\ \bibnamefont {Bardin}}, \bibinfo
  {author} {\bibfnamefont {R.}~\bibnamefont {Barends}}, \bibinfo {author}
  {\bibfnamefont {J.}~\bibnamefont {Basso}}, \bibinfo {author} {\bibfnamefont
  {A.}~\bibnamefont {Bengtsson}}, \bibinfo {author} {\bibfnamefont
  {A.}~\bibnamefont {Bilmes}}, \bibinfo {author} {\bibfnamefont
  {M.}~\bibnamefont {Broughton}}, \bibinfo {author} {\bibfnamefont {B.~B.}\
  \bibnamefont {Buckley}}, \bibinfo {author} {\bibfnamefont {D.~A.}\
  \bibnamefont {Buell}}, \bibinfo {author} {\bibfnamefont {B.}~\bibnamefont
  {Burkett}}, \bibinfo {author} {\bibfnamefont {N.}~\bibnamefont {Bushnell}},
  \bibinfo {author} {\bibfnamefont {B.}~\bibnamefont {Chiaro}}, \bibinfo
  {author} {\bibfnamefont {R.}~\bibnamefont {Collins}}, \bibinfo {author}
  {\bibfnamefont {W.}~\bibnamefont {Courtney}}, \bibinfo {author}
  {\bibfnamefont {S.}~\bibnamefont {Demura}}, \bibinfo {author} {\bibfnamefont
  {A.~R.}\ \bibnamefont {Derk}}, \bibinfo {author} {\bibfnamefont
  {D.}~\bibnamefont {Eppens}}, \bibinfo {author} {\bibfnamefont
  {C.}~\bibnamefont {Erickson}}, \bibinfo {author} {\bibfnamefont
  {L.}~\bibnamefont {Faoro}}, \bibinfo {author} {\bibfnamefont
  {E.}~\bibnamefont {Farhi}}, \bibinfo {author} {\bibfnamefont {A.~G.}\
  \bibnamefont {Fowler}}, \bibinfo {author} {\bibfnamefont {B.}~\bibnamefont
  {Foxen}}, \bibinfo {author} {\bibfnamefont {M.}~\bibnamefont {Giustina}},
  \bibinfo {author} {\bibfnamefont {A.}~\bibnamefont {Greene}}, \bibinfo
  {author} {\bibfnamefont {J.~A.}\ \bibnamefont {Gross}}, \bibinfo {author}
  {\bibfnamefont {M.~P.}\ \bibnamefont {Harrigan}}, \bibinfo {author}
  {\bibfnamefont {S.~D.}\ \bibnamefont {Harrington}}, \bibinfo {author}
  {\bibfnamefont {J.}~\bibnamefont {Hilton}}, \bibinfo {author} {\bibfnamefont
  {S.}~\bibnamefont {Hong}}, \bibinfo {author} {\bibfnamefont {T.}~\bibnamefont
  {Huang}}, \bibinfo {author} {\bibfnamefont {W.~J.}\ \bibnamefont {Huggins}},
  \bibinfo {author} {\bibfnamefont {L.~B.}\ \bibnamefont {Ioffe}}, \bibinfo
  {author} {\bibfnamefont {S.~V.}\ \bibnamefont {Isakov}}, \bibinfo {author}
  {\bibfnamefont {E.}~\bibnamefont {Jeffrey}}, \bibinfo {author} {\bibfnamefont
  {Z.}~\bibnamefont {Jiang}}, \bibinfo {author} {\bibfnamefont
  {D.}~\bibnamefont {Kafri}}, \bibinfo {author} {\bibfnamefont
  {K.}~\bibnamefont {Kechedzhi}}, \bibinfo {author} {\bibfnamefont
  {T.}~\bibnamefont {Khattar}}, \bibinfo {author} {\bibfnamefont
  {S.}~\bibnamefont {Kim}}, \bibinfo {author} {\bibfnamefont {P.~V.}\
  \bibnamefont {Klimov}}, \bibinfo {author} {\bibfnamefont {A.~N.}\
  \bibnamefont {Korotkov}}, \bibinfo {author} {\bibfnamefont {F.}~\bibnamefont
  {Kostritsa}}, \bibinfo {author} {\bibfnamefont {D.}~\bibnamefont {Landhuis}},
  \bibinfo {author} {\bibfnamefont {P.}~\bibnamefont {Laptev}}, \bibinfo
  {author} {\bibfnamefont {A.}~\bibnamefont {Locharla}}, \bibinfo {author}
  {\bibfnamefont {E.}~\bibnamefont {Lucero}}, \bibinfo {author} {\bibfnamefont
  {O.}~\bibnamefont {Martin}}, \bibinfo {author} {\bibfnamefont {J.~R.}\
  \bibnamefont {McClean}}, \bibinfo {author} {\bibfnamefont {M.}~\bibnamefont
  {McEwen}}, \bibinfo {author} {\bibfnamefont {K.~C.}\ \bibnamefont {Miao}},
  \bibinfo {author} {\bibfnamefont {M.}~\bibnamefont {Mohseni}}, \bibinfo
  {author} {\bibfnamefont {S.}~\bibnamefont {Montazeri}}, \bibinfo {author}
  {\bibfnamefont {W.}~\bibnamefont {Mruczkiewicz}}, \bibinfo {author}
  {\bibfnamefont {J.}~\bibnamefont {Mutus}}, \bibinfo {author} {\bibfnamefont
  {O.}~\bibnamefont {Naaman}}, \bibinfo {author} {\bibfnamefont
  {M.}~\bibnamefont {Neeley}}, \bibinfo {author} {\bibfnamefont
  {C.}~\bibnamefont {Neill}}, \bibinfo {author} {\bibfnamefont {M.~Y.}\
  \bibnamefont {Niu}}, \bibinfo {author} {\bibfnamefont {T.~E.}\ \bibnamefont
  {O'Brien}}, \bibinfo {author} {\bibfnamefont {A.}~\bibnamefont {Opremcak}},
  \bibinfo {author} {\bibfnamefont {B.}~\bibnamefont {Pat{\'o}}}, \bibinfo
  {author} {\bibfnamefont {A.}~\bibnamefont {Petukhov}}, \bibinfo {author}
  {\bibfnamefont {N.~C.}\ \bibnamefont {Rubin}}, \bibinfo {author}
  {\bibfnamefont {D.}~\bibnamefont {Sank}}, \bibinfo {author} {\bibfnamefont
  {V.}~\bibnamefont {Shvarts}}, \bibinfo {author} {\bibfnamefont
  {D.}~\bibnamefont {Strain}}, \bibinfo {author} {\bibfnamefont
  {M.}~\bibnamefont {Szalay}}, \bibinfo {author} {\bibfnamefont
  {B.}~\bibnamefont {Villalonga}}, \bibinfo {author} {\bibfnamefont {T.~C.}\
  \bibnamefont {White}}, \bibinfo {author} {\bibfnamefont {Z.}~\bibnamefont
  {Yao}}, \bibinfo {author} {\bibfnamefont {P.}~\bibnamefont {Yeh}}, \bibinfo
  {author} {\bibfnamefont {J.}~\bibnamefont {Yoo}}, \bibinfo {author}
  {\bibfnamefont {A.}~\bibnamefont {Zalcman}}, \bibinfo {author} {\bibfnamefont
  {H.}~\bibnamefont {Neven}}, \bibinfo {author} {\bibfnamefont
  {S.}~\bibnamefont {Boixo}}, \bibinfo {author} {\bibfnamefont
  {A.}~\bibnamefont {Megrant}}, \bibinfo {author} {\bibfnamefont
  {Y.}~\bibnamefont {Chen}}, \bibinfo {author} {\bibfnamefont {J.}~\bibnamefont
  {Kelly}}, \bibinfo {author} {\bibfnamefont {V.}~\bibnamefont {Smelyanskiy}},
  \bibinfo {author} {\bibfnamefont {A.~Y.}\ \bibnamefont {Kitaev}}, \bibinfo
  {author} {\bibfnamefont {M.}~\bibnamefont {Knap}}, \bibinfo {author}
  {\bibfnamefont {F.}~\bibnamefont {Pollmann}},\ and\ \bibinfo {author}
  {\bibfnamefont {P.}~\bibnamefont {Roushan}},\ }\bibfield  {title} {\bibinfo
  {title} {{Realizing topologically ordered states on a quantum processor}},\
  }\href {https://doi.org/10.1126/science.abi8378} {\bibfield  {journal}
  {\bibinfo  {journal} {Science}\ }\textbf {\bibinfo {volume} {374}},\ \bibinfo
  {pages} {1237} (\bibinfo {year} {2021})}\BibitemShut {NoStop}%
\bibitem [{\citenamefont {Qi}\ and\ \citenamefont
  {Ranard}(2019)}]{Qi:2019CorrMat}%
  \BibitemOpen
  \bibfield  {author} {\bibinfo {author} {\bibfnamefont {X.-L.}\ \bibnamefont
  {Qi}}\ and\ \bibinfo {author} {\bibfnamefont {D.}~\bibnamefont {Ranard}},\
  }\bibfield  {title} {\bibinfo {title} {{Determining a local Hamiltonian from
  a single eigenstate}},\ }\href {https://doi.org/10.22331/q-2019-07-08-159}
  {\bibfield  {journal} {\bibinfo  {journal} {Quantum}\ }\textbf {\bibinfo
  {volume} {3}},\ \bibinfo {pages} {159} (\bibinfo {year} {2019})}\BibitemShut
  {NoStop}%
\bibitem [{\citenamefont {Garrison}\ and\ \citenamefont
  {Grover}(2018)}]{Garrison:2018da}%
  \BibitemOpen
  \bibfield  {author} {\bibinfo {author} {\bibfnamefont {J.~R.}\ \bibnamefont
  {Garrison}}\ and\ \bibinfo {author} {\bibfnamefont {T.}~\bibnamefont
  {Grover}},\ }\bibfield  {title} {\bibinfo {title} {{Does a Single Eigenstate
  Encode the Full Hamiltonian?}},\ }\href
  {https://doi.org/10.1103/PhysRevX.8.021026} {\bibfield  {journal} {\bibinfo
  {journal} {Phys. Rev. X}\ }\textbf {\bibinfo {volume} {8}},\ \bibinfo {pages}
  {021026} (\bibinfo {year} {2018})}\BibitemShut {NoStop}%
\bibitem [{Note1()}]{Note1}%
  \BibitemOpen
  \bibinfo {note} {Refer to Appendix A for a visual demonstration.}\BibitemShut
  {Stop}%
\bibitem [{Note2()}]{Note2}%
  \BibitemOpen
  \bibinfo {note} {Refer to Appendix D for details.}\BibitemShut {Stop}%
\bibitem [{Note3()}]{Note3}%
  \BibitemOpen
  \bibinfo {note} {See Appendix B for a proof}\BibitemShut {NoStop}%
\bibitem [{Note4()}]{Note4}%
  \BibitemOpen
  \bibinfo {note} {This is proved in Appendix C.}\BibitemShut {Stop}%
\bibitem [{\citenamefont {Moudgalya}\ \emph {et~al.}(2018)\citenamefont
  {Moudgalya}, \citenamefont {Regnault},\ and\ \citenamefont
  {Bernevig}}]{Moudgalya:2018AKLT}%
  \BibitemOpen
  \bibfield  {author} {\bibinfo {author} {\bibfnamefont {S.}~\bibnamefont
  {Moudgalya}}, \bibinfo {author} {\bibfnamefont {N.}~\bibnamefont
  {Regnault}},\ and\ \bibinfo {author} {\bibfnamefont {B.~A.}\ \bibnamefont
  {Bernevig}},\ }\bibfield  {title} {\bibinfo {title} {{Entanglement of exact
  excited states of Affleck-Kennedy-Lieb-Tasaki models: Exact results,
  many-body scars, and violation of the strong eigenstate thermalization
  hypothesis}},\ }\href {https://doi.org/10.1103/PhysRevB.98.235156} {\bibfield
   {journal} {\bibinfo  {journal} {Phys. Rev. B}\ }\textbf {\bibinfo {volume}
  {98}},\ \bibinfo {pages} {235156} (\bibinfo {year} {2018})}\BibitemShut
  {NoStop}%
\bibitem [{\citenamefont {Mark}\ \emph {et~al.}(2020)\citenamefont {Mark},
  \citenamefont {Lin},\ and\ \citenamefont {Motrunich}}]{Mark:2020AKLT}%
  \BibitemOpen
  \bibfield  {author} {\bibinfo {author} {\bibfnamefont {D.~K.}\ \bibnamefont
  {Mark}}, \bibinfo {author} {\bibfnamefont {C.-J.}\ \bibnamefont {Lin}},\ and\
  \bibinfo {author} {\bibfnamefont {O.~I.}\ \bibnamefont {Motrunich}},\
  }\bibfield  {title} {\bibinfo {title} {{Unified structure for exact towers of
  scar states in the Affleck-Kennedy-Lieb-Tasaki and other models}},\ }\href
  {https://doi.org/10.1103/PhysRevB.101.195131} {\bibfield  {journal} {\bibinfo
   {journal} {Phys. Rev. B}\ }\textbf {\bibinfo {volume} {101}},\ \bibinfo
  {pages} {195131} (\bibinfo {year} {2020})}\BibitemShut {NoStop}%
\bibitem [{Note5()}]{Note5}%
  \BibitemOpen
  \bibinfo {note} {Refer to Appendix E for details.}\BibitemShut {Stop}%
\bibitem [{Note6()}]{Note6}%
  \BibitemOpen
  \bibinfo {note} {Similar results are also obtained for scar states in the
  extended Fermi-Hubbard model, see Appendix F that also includes Refs. \cite
  {Mark:2020ep,Moudgalya:2020ep,Yang:1989PRL_Hubbard,Vafek:2017bv}.}\BibitemShut
  {Stop}%
\bibitem [{\citenamefont {Gottesman}(1997)}]{Gottesman:1997sc}%
  \BibitemOpen
  \bibfield  {author} {\bibinfo {author} {\bibfnamefont {D.}~\bibnamefont
  {Gottesman}},\ }\href {https://arxiv.org/abs/quant-ph/9705052v1} {\bibinfo
  {title} {{Stabilizer Codes and Quantum Error Correction}}} (\bibinfo {year}
  {1997}),\ \Eprint {https://arxiv.org/abs/quant-ph/9705052v1}
  {arXiv:quant-ph/9705052v1} \BibitemShut {NoStop}%
\bibitem [{\citenamefont {Gottesman}(1998)}]{Gottesman:1998to}%
  \BibitemOpen
  \bibfield  {author} {\bibinfo {author} {\bibfnamefont {D.}~\bibnamefont
  {Gottesman}},\ }\bibfield  {title} {\bibinfo {title} {{Theory of
  fault-tolerant quantum computation}},\ }\href
  {https://doi.org/10.1103/PhysRevA.57.127} {\bibfield  {journal} {\bibinfo
  {journal} {Phys. Rev. A}\ }\textbf {\bibinfo {volume} {57}},\ \bibinfo
  {pages} {127} (\bibinfo {year} {1998})}\BibitemShut {NoStop}%
\bibitem [{\citenamefont {Pachos}(2012)}]{Pachos:Book}%
  \BibitemOpen
  \bibfield  {author} {\bibinfo {author} {\bibfnamefont {J.~K.}\ \bibnamefont
  {Pachos}},\ }\href {https://doi.org/10.1017/CBO9780511792908} {\emph
  {\bibinfo {title} {{Introduction to Topological Quantum Computation}}}}\
  (\bibinfo  {publisher} {Cambridge University Press},\ \bibinfo {address}
  {Cambridge},\ \bibinfo {year} {2012})\BibitemShut {NoStop}%
\bibitem [{\citenamefont {Kitaev}(2006)}]{Kitaev:2006ai}%
  \BibitemOpen
  \bibfield  {author} {\bibinfo {author} {\bibfnamefont {A.~Y.}\ \bibnamefont
  {Kitaev}},\ }\bibfield  {title} {\bibinfo {title} {{Anyons in an exactly
  solved model and beyond}},\ }\href
  {https://doi.org/10.1016/j.aop.2005.10.005} {\bibfield  {journal} {\bibinfo
  {journal} {Ann. Phys.}\ }\textbf {\bibinfo {volume} {321}},\ \bibinfo {pages}
  {2} (\bibinfo {year} {2006})}\BibitemShut {NoStop}%
\bibitem [{\citenamefont {Sakurai}(1994)}]{Sakurai:Book}%
  \BibitemOpen
  \bibfield  {author} {\bibinfo {author} {\bibfnamefont {J.~J.}\ \bibnamefont
  {Sakurai}},\ }\href@noop {} {\emph {\bibinfo {title} {{Modern Quantum
  Mechanics; Rev. Ed.}}}}\ (\bibinfo  {publisher} {Addison-Wesley},\ \bibinfo
  {address} {Reading, MA},\ \bibinfo {year} {1994})\BibitemShut {NoStop}%
\bibitem [{\citenamefont {Bernien}\ \emph {et~al.}(2017)\citenamefont
  {Bernien}, \citenamefont {Schwartz}, \citenamefont {Keesling}, \citenamefont
  {Levine}, \citenamefont {Omran}, \citenamefont {Pichler}, \citenamefont
  {Choi}, \citenamefont {Zibrov}, \citenamefont {Endres}, \citenamefont
  {Greiner}, \citenamefont {Vuleti{\'c}},\ and\ \citenamefont
  {Lukin}}]{Bernien:2017pm}%
  \BibitemOpen
  \bibfield  {author} {\bibinfo {author} {\bibfnamefont {H.}~\bibnamefont
  {Bernien}}, \bibinfo {author} {\bibfnamefont {S.}~\bibnamefont {Schwartz}},
  \bibinfo {author} {\bibfnamefont {A.}~\bibnamefont {Keesling}}, \bibinfo
  {author} {\bibfnamefont {H.}~\bibnamefont {Levine}}, \bibinfo {author}
  {\bibfnamefont {A.}~\bibnamefont {Omran}}, \bibinfo {author} {\bibfnamefont
  {H.}~\bibnamefont {Pichler}}, \bibinfo {author} {\bibfnamefont
  {S.}~\bibnamefont {Choi}}, \bibinfo {author} {\bibfnamefont {A.~S.}\
  \bibnamefont {Zibrov}}, \bibinfo {author} {\bibfnamefont {M.}~\bibnamefont
  {Endres}}, \bibinfo {author} {\bibfnamefont {M.}~\bibnamefont {Greiner}},
  \bibinfo {author} {\bibfnamefont {V.}~\bibnamefont {Vuleti{\'c}}},\ and\
  \bibinfo {author} {\bibfnamefont {M.~D.}\ \bibnamefont {Lukin}},\ }\bibfield
  {title} {\bibinfo {title} {{Probing many-body dynamics on a 51-atom quantum
  simulator}},\ }\href {https://doi.org/10.1038/nature24622} {\bibfield
  {journal} {\bibinfo  {journal} {Nature}\ }\textbf {\bibinfo {volume} {551}},\
  \bibinfo {pages} {579} (\bibinfo {year} {2017})}\BibitemShut {NoStop}%
\bibitem [{\citenamefont {Turner}\ \emph
  {et~al.}(2018{\natexlab{a}})\citenamefont {Turner}, \citenamefont
  {Michailidis}, \citenamefont {Abanin}, \citenamefont {Serbyn},\ and\
  \citenamefont {Papi{\'c}}}]{Turner:2018wg}%
  \BibitemOpen
  \bibfield  {author} {\bibinfo {author} {\bibfnamefont {C.~J.}\ \bibnamefont
  {Turner}}, \bibinfo {author} {\bibfnamefont {A.~A.}\ \bibnamefont
  {Michailidis}}, \bibinfo {author} {\bibfnamefont {D.~A.}\ \bibnamefont
  {Abanin}}, \bibinfo {author} {\bibfnamefont {M.}~\bibnamefont {Serbyn}},\
  and\ \bibinfo {author} {\bibfnamefont {Z.}~\bibnamefont {Papi{\'c}}},\
  }\bibfield  {title} {\bibinfo {title} {{Weak ergodicity breaking from quantum
  many-body scars}},\ }\href {https://doi.org/10.1038/s41567-018-0137-5}
  {\bibfield  {journal} {\bibinfo  {journal} {Nat. Phys.}\ }\textbf {\bibinfo
  {volume} {14}},\ \bibinfo {pages} {745} (\bibinfo {year}
  {2018}{\natexlab{a}})}\BibitemShut {NoStop}%
\bibitem [{\citenamefont {Turner}\ \emph
  {et~al.}(2018{\natexlab{b}})\citenamefont {Turner}, \citenamefont
  {Michailidis}, \citenamefont {Abanin}, \citenamefont {Serbyn},\ and\
  \citenamefont {Papi{\'c}}}]{Turner:2018qs}%
  \BibitemOpen
  \bibfield  {author} {\bibinfo {author} {\bibfnamefont {C.~J.}\ \bibnamefont
  {Turner}}, \bibinfo {author} {\bibfnamefont {A.~A.}\ \bibnamefont
  {Michailidis}}, \bibinfo {author} {\bibfnamefont {D.~A.}\ \bibnamefont
  {Abanin}}, \bibinfo {author} {\bibfnamefont {M.}~\bibnamefont {Serbyn}},\
  and\ \bibinfo {author} {\bibfnamefont {Z.}~\bibnamefont {Papi{\'c}}},\
  }\bibfield  {title} {\bibinfo {title} {{Quantum scarred eigenstates in a
  Rydberg atom chain: Entanglement, breakdown of thermalization, and stability
  to perturbations}},\ }\href {https://doi.org/10.1103/PhysRevB.98.155134}
  {\bibfield  {journal} {\bibinfo  {journal} {Phys. Rev. B}\ }\textbf {\bibinfo
  {volume} {98}},\ \bibinfo {pages} {155134} (\bibinfo {year}
  {2018}{\natexlab{b}})}\BibitemShut {NoStop}%
\bibitem [{Note7()}]{Note7}%
  \BibitemOpen
  \bibinfo {note} {See Appendix G for details}\BibitemShut {NoStop}%
\bibitem [{Note8()}]{Note8}%
  \BibitemOpen
  \bibinfo {note} {See Appendix H for an experimental protocol in cold atom
  systems, which also includes Refs.~\cite
  {Bakr:2009aq,Sherson:2010sa,Weitenberg:2011ss,Cheuk:2015qg,Parsons:2015sr,Gross:2021qg,Wang:2015ca,Xia:2015rb,Wang:2016sq}}\BibitemShut
  {NoStop}%
\bibitem [{\citenamefont {Ohliger}\ \emph {et~al.}(2013)\citenamefont
  {Ohliger}, \citenamefont {Nesme},\ and\ \citenamefont
  {Eisert}}]{Ohliger:2013ea}%
  \BibitemOpen
  \bibfield  {author} {\bibinfo {author} {\bibfnamefont {M.}~\bibnamefont
  {Ohliger}}, \bibinfo {author} {\bibfnamefont {V.}~\bibnamefont {Nesme}},\
  and\ \bibinfo {author} {\bibfnamefont {J.}~\bibnamefont {Eisert}},\
  }\bibfield  {title} {\bibinfo {title} {Efficient and feasible state
  tomography of quantum many-body systems},\ }\href
  {https://doi.org/10.1088/1367-2630/15/1/015024} {\bibfield  {journal}
  {\bibinfo  {journal} {New J. Phys.}\ }\textbf {\bibinfo {volume} {15}},\
  \bibinfo {pages} {015024} (\bibinfo {year} {2013})}\BibitemShut {NoStop}%
\bibitem [{\citenamefont {Lanyon}\ \emph {et~al.}(2017)\citenamefont {Lanyon},
  \citenamefont {Maier}, \citenamefont {Holz{\"a}pfel}, \citenamefont
  {Baumgratz}, \citenamefont {Hempel}, \citenamefont {Jurcevic}, \citenamefont
  {Dhand}, \citenamefont {Buyskikh}, \citenamefont {Daley}, \citenamefont
  {Cramer}, \citenamefont {Plenio}, \citenamefont {Blatt},\ and\ \citenamefont
  {Roos}}]{Lanyon:2017eo}%
  \BibitemOpen
  \bibfield  {author} {\bibinfo {author} {\bibfnamefont {B.~P.}\ \bibnamefont
  {Lanyon}}, \bibinfo {author} {\bibfnamefont {C.}~\bibnamefont {Maier}},
  \bibinfo {author} {\bibfnamefont {M.}~\bibnamefont {Holz{\"a}pfel}}, \bibinfo
  {author} {\bibfnamefont {T.}~\bibnamefont {Baumgratz}}, \bibinfo {author}
  {\bibfnamefont {C.}~\bibnamefont {Hempel}}, \bibinfo {author} {\bibfnamefont
  {P.}~\bibnamefont {Jurcevic}}, \bibinfo {author} {\bibfnamefont
  {I.}~\bibnamefont {Dhand}}, \bibinfo {author} {\bibfnamefont {A.~S.}\
  \bibnamefont {Buyskikh}}, \bibinfo {author} {\bibfnamefont {A.~J.}\
  \bibnamefont {Daley}}, \bibinfo {author} {\bibfnamefont {M.}~\bibnamefont
  {Cramer}}, \bibinfo {author} {\bibfnamefont {M.~B.}\ \bibnamefont {Plenio}},
  \bibinfo {author} {\bibfnamefont {R.}~\bibnamefont {Blatt}},\ and\ \bibinfo
  {author} {\bibfnamefont {C.~F.}\ \bibnamefont {Roos}},\ }\bibfield  {title}
  {\bibinfo {title} {{Efficient tomography of a quantum many-body~system}},\
  }\href {https://doi.org/10.1038/nphys4244} {\bibfield  {journal} {\bibinfo
  {journal} {Nat. Phys.}\ }\textbf {\bibinfo {volume} {13}},\ \bibinfo {pages}
  {1158} (\bibinfo {year} {2017})}\BibitemShut {NoStop}%
\bibitem [{\citenamefont {Dalmonte}\ \emph {et~al.}(2018)\citenamefont
  {Dalmonte}, \citenamefont {Vermersch},\ and\ \citenamefont
  {Zoller}}]{Dalmonte:2018qs}%
  \BibitemOpen
  \bibfield  {author} {\bibinfo {author} {\bibfnamefont {M.}~\bibnamefont
  {Dalmonte}}, \bibinfo {author} {\bibfnamefont {B.}~\bibnamefont
  {Vermersch}},\ and\ \bibinfo {author} {\bibfnamefont {P.}~\bibnamefont
  {Zoller}},\ }\bibfield  {title} {\bibinfo {title} {Quantum simulation and
  spectroscopy of entanglement {{Hamiltonians}}},\ }\href
  {https://doi.org/10.1038/s41567-018-0151-7} {\bibfield  {journal} {\bibinfo
  {journal} {Nat. Phys.}\ }\textbf {\bibinfo {volume} {14}},\ \bibinfo {pages}
  {827} (\bibinfo {year} {2018})}\BibitemShut {NoStop}%
\bibitem [{\citenamefont {Kokail}\ \emph {et~al.}(2021)\citenamefont {Kokail},
  \citenamefont {van Bijnen}, \citenamefont {Elben}, \citenamefont
  {Vermersch},\ and\ \citenamefont {Zoller}}]{Kokail:2021dc}%
  \BibitemOpen
  \bibfield  {author} {\bibinfo {author} {\bibfnamefont {C.}~\bibnamefont
  {Kokail}}, \bibinfo {author} {\bibfnamefont {R.}~\bibnamefont {van Bijnen}},
  \bibinfo {author} {\bibfnamefont {A.}~\bibnamefont {Elben}}, \bibinfo
  {author} {\bibfnamefont {B.}~\bibnamefont {Vermersch}},\ and\ \bibinfo
  {author} {\bibfnamefont {P.}~\bibnamefont {Zoller}},\ }\bibfield  {title}
  {\bibinfo {title} {{Entanglement Hamiltonian tomography in quantum
  simulation}},\ }\href {https://doi.org/10.1038/s41567-021-01260-w} {\bibfield
   {journal} {\bibinfo  {journal} {Nat. Phys.}\ }\textbf {\bibinfo {volume}
  {17}},\ \bibinfo {pages} {936} (\bibinfo {year} {2021})}\BibitemShut
  {NoStop}%
\bibitem [{\citenamefont {Preskill}(2018)}]{Preskill:2018gt}%
  \BibitemOpen
  \bibfield  {author} {\bibinfo {author} {\bibfnamefont {J.}~\bibnamefont
  {Preskill}},\ }\bibfield  {title} {\bibinfo {title} {{Quantum Computing in
  the NISQ era and beyond}},\ }\href {https://doi.org/10.22331/q-2018-08-06-79}
  {\bibfield  {journal} {\bibinfo  {journal} {Quantum}\ }\textbf {\bibinfo
  {volume} {2}},\ \bibinfo {pages} {79} (\bibinfo {year} {2018})}\BibitemShut
  {NoStop}%
\bibitem [{\citenamefont {Auerbach}(1994)}]{Auerbach:Book}%
  \BibitemOpen
  \bibfield  {author} {\bibinfo {author} {\bibfnamefont {A.}~\bibnamefont
  {Auerbach}},\ }\href {https://doi.org/10.1007/978-1-4612-0869-3} {\emph
  {\bibinfo {title} {{Interacting Electrons and Quantum Magnetism}}}}\
  (\bibinfo  {publisher} {Springer},\ \bibinfo {address} {New York},\ \bibinfo
  {year} {1994})\BibitemShut {NoStop}%
\bibitem [{\citenamefont {Mark}\ and\ \citenamefont
  {Motrunich}(2020)}]{Mark:2020ep}%
  \BibitemOpen
  \bibfield  {author} {\bibinfo {author} {\bibfnamefont {D.~K.}\ \bibnamefont
  {Mark}}\ and\ \bibinfo {author} {\bibfnamefont {O.~I.}\ \bibnamefont
  {Motrunich}},\ }\bibfield  {title} {\bibinfo {title} {{$\eta$-pairing states
  as true scars in an extended Hubbard model}},\ }\href
  {https://doi.org/10.1103/PhysRevB.102.075132} {\bibfield  {journal} {\bibinfo
   {journal} {Phys. Rev. B}\ }\textbf {\bibinfo {volume} {102}},\ \bibinfo
  {pages} {075132} (\bibinfo {year} {2020})}\BibitemShut {NoStop}%
\bibitem [{\citenamefont {Moudgalya}\ \emph
  {et~al.}(2020{\natexlab{a}})\citenamefont {Moudgalya}, \citenamefont
  {Regnault},\ and\ \citenamefont {Bernevig}}]{Moudgalya:2020Hubbard}%
  \BibitemOpen
  \bibfield  {author} {\bibinfo {author} {\bibfnamefont {S.}~\bibnamefont
  {Moudgalya}}, \bibinfo {author} {\bibfnamefont {N.}~\bibnamefont
  {Regnault}},\ and\ \bibinfo {author} {\bibfnamefont {B.~A.}\ \bibnamefont
  {Bernevig}},\ }\bibfield  {title} {\bibinfo {title} {{$\eta$-pairing in
  Hubbard models: From spectrum generating algebras to quantum many-body
  scars}},\ }\href {https://doi.org/10.1103/PhysRevB.102.085140} {\bibfield
  {journal} {\bibinfo  {journal} {Phys. Rev. B}\ }\textbf {\bibinfo {volume}
  {102}},\ \bibinfo {pages} {085140} (\bibinfo {year}
  {2020}{\natexlab{a}})}\BibitemShut {NoStop}%
\bibitem [{\citenamefont {Yang}(1989)}]{Yang:1989PRL_Hubbard}%
  \BibitemOpen
  \bibfield  {author} {\bibinfo {author} {\bibfnamefont {C.~N.}\ \bibnamefont
  {Yang}},\ }\bibfield  {title} {\bibinfo {title} {{$\eta$ pairing and
  off-diagonal long-range order in a Hubbard model}},\ }\href
  {https://doi.org/10.1103/PhysRevLett.63.2144} {\bibfield  {journal} {\bibinfo
   {journal} {Phys. Rev. Lett.}\ }\textbf {\bibinfo {volume} {63}},\ \bibinfo
  {pages} {2144} (\bibinfo {year} {1989})}\BibitemShut {NoStop}%
\bibitem [{\citenamefont {Vafek}\ \emph {et~al.}(2017)\citenamefont {Vafek},
  \citenamefont {Regnault},\ and\ \citenamefont {Bernevig}}]{Vafek:2017bv}%
  \BibitemOpen
  \bibfield  {author} {\bibinfo {author} {\bibfnamefont {O.}~\bibnamefont
  {Vafek}}, \bibinfo {author} {\bibfnamefont {N.}~\bibnamefont {Regnault}},\
  and\ \bibinfo {author} {\bibfnamefont {B.~A.}\ \bibnamefont {Bernevig}},\
  }\bibfield  {title} {\bibinfo {title} {{Entanglement of exact excited
  eigenstates of the Hubbard model in arbitrary dimension}},\ }\href
  {https://doi.org/10.21468/SciPostPhys.3.6.043} {\bibfield  {journal}
  {\bibinfo  {journal} {SciPost Phys.}\ }\textbf {\bibinfo {volume} {3}},\
  \bibinfo {pages} {043} (\bibinfo {year} {2017})}\BibitemShut {NoStop}%
\bibitem [{Yao()}]{Yao:source_code}%
  \BibitemOpen
  \href@noop {} {\bibinfo {title} {Our open source code for the program
  including the original data are freely available at
  \href{https://github.com/ZhiyuanYao/CorrMat_EE.git}{https://github.com/{ZhiyuanYao}/{CorrMat}\_{EE}.git}}}\BibitemShut
  {NoStop}%
\bibitem [{\citenamefont {Bakr}\ \emph {et~al.}(2009)\citenamefont {Bakr},
  \citenamefont {Gillen}, \citenamefont {Peng}, \citenamefont {F{\"o}lling},\
  and\ \citenamefont {Greiner}}]{Bakr:2009aq}%
  \BibitemOpen
  \bibfield  {author} {\bibinfo {author} {\bibfnamefont {W.~S.}\ \bibnamefont
  {Bakr}}, \bibinfo {author} {\bibfnamefont {J.~I.}\ \bibnamefont {Gillen}},
  \bibinfo {author} {\bibfnamefont {A.}~\bibnamefont {Peng}}, \bibinfo {author}
  {\bibfnamefont {S.}~\bibnamefont {F{\"o}lling}},\ and\ \bibinfo {author}
  {\bibfnamefont {M.}~\bibnamefont {Greiner}},\ }\bibfield  {title} {\bibinfo
  {title} {{A quantum gas microscope for detecting single atoms in a
  Hubbard-regime optical lattice}},\ }\href
  {https://doi.org/10.1038/nature08482} {\bibfield  {journal} {\bibinfo
  {journal} {Nature}\ }\textbf {\bibinfo {volume} {462}},\ \bibinfo {pages}
  {74} (\bibinfo {year} {2009})}\BibitemShut {NoStop}%
\bibitem [{\citenamefont {Sherson}\ \emph {et~al.}(2010)\citenamefont
  {Sherson}, \citenamefont {Weitenberg}, \citenamefont {Endres}, \citenamefont
  {Cheneau}, \citenamefont {Bloch},\ and\ \citenamefont
  {Kuhr}}]{Sherson:2010sa}%
  \BibitemOpen
  \bibfield  {author} {\bibinfo {author} {\bibfnamefont {J.~F.}\ \bibnamefont
  {Sherson}}, \bibinfo {author} {\bibfnamefont {C.}~\bibnamefont {Weitenberg}},
  \bibinfo {author} {\bibfnamefont {M.}~\bibnamefont {Endres}}, \bibinfo
  {author} {\bibfnamefont {M.}~\bibnamefont {Cheneau}}, \bibinfo {author}
  {\bibfnamefont {I.}~\bibnamefont {Bloch}},\ and\ \bibinfo {author}
  {\bibfnamefont {S.}~\bibnamefont {Kuhr}},\ }\bibfield  {title} {\bibinfo
  {title} {{Single-atom-resolved fluorescence imaging of an atomic Mott
  insulator}},\ }\href {https://doi.org/10.1038/nature09378} {\bibfield
  {journal} {\bibinfo  {journal} {Nature}\ }\textbf {\bibinfo {volume} {467}},\
  \bibinfo {pages} {68} (\bibinfo {year} {2010})}\BibitemShut {NoStop}%
\bibitem [{\citenamefont {Weitenberg}\ \emph {et~al.}(2011)\citenamefont
  {Weitenberg}, \citenamefont {Endres}, \citenamefont {Sherson}, \citenamefont
  {Cheneau}, \citenamefont {Schau{\ss}}, \citenamefont {Fukuhara},
  \citenamefont {Bloch},\ and\ \citenamefont {Kuhr}}]{Weitenberg:2011ss}%
  \BibitemOpen
  \bibfield  {author} {\bibinfo {author} {\bibfnamefont {C.}~\bibnamefont
  {Weitenberg}}, \bibinfo {author} {\bibfnamefont {M.}~\bibnamefont {Endres}},
  \bibinfo {author} {\bibfnamefont {J.~F.}\ \bibnamefont {Sherson}}, \bibinfo
  {author} {\bibfnamefont {M.}~\bibnamefont {Cheneau}}, \bibinfo {author}
  {\bibfnamefont {P.}~\bibnamefont {Schau{\ss}}}, \bibinfo {author}
  {\bibfnamefont {T.}~\bibnamefont {Fukuhara}}, \bibinfo {author}
  {\bibfnamefont {I.}~\bibnamefont {Bloch}},\ and\ \bibinfo {author}
  {\bibfnamefont {S.}~\bibnamefont {Kuhr}},\ }\bibfield  {title} {\bibinfo
  {title} {{Single-spin addressing in an atomic Mott insulator}},\ }\href
  {https://doi.org/10.1038/nature09827} {\bibfield  {journal} {\bibinfo
  {journal} {Nature}\ }\textbf {\bibinfo {volume} {471}},\ \bibinfo {pages}
  {319} (\bibinfo {year} {2011})}\BibitemShut {NoStop}%
\bibitem [{\citenamefont {Cheuk}\ \emph {et~al.}(2015)\citenamefont {Cheuk},
  \citenamefont {Nichols}, \citenamefont {Okan}, \citenamefont {Gersdorf},
  \citenamefont {Ramasesh}, \citenamefont {Bakr}, \citenamefont {Lompe},\ and\
  \citenamefont {Zwierlein}}]{Cheuk:2015qg}%
  \BibitemOpen
  \bibfield  {author} {\bibinfo {author} {\bibfnamefont {L.~W.}\ \bibnamefont
  {Cheuk}}, \bibinfo {author} {\bibfnamefont {M.~A.}\ \bibnamefont {Nichols}},
  \bibinfo {author} {\bibfnamefont {M.}~\bibnamefont {Okan}}, \bibinfo {author}
  {\bibfnamefont {T.}~\bibnamefont {Gersdorf}}, \bibinfo {author}
  {\bibfnamefont {V.~V.}\ \bibnamefont {Ramasesh}}, \bibinfo {author}
  {\bibfnamefont {W.~S.}\ \bibnamefont {Bakr}}, \bibinfo {author}
  {\bibfnamefont {T.}~\bibnamefont {Lompe}},\ and\ \bibinfo {author}
  {\bibfnamefont {M.~W.}\ \bibnamefont {Zwierlein}},\ }\bibfield  {title}
  {\bibinfo {title} {{Quantum-Gas Microscope for Fermionic Atoms}},\ }\href
  {https://doi.org/10.1103/PhysRevLett.114.193001} {\bibfield  {journal}
  {\bibinfo  {journal} {Phys. Rev. Lett.}\ }\textbf {\bibinfo {volume} {114}},\
  \bibinfo {pages} {193001} (\bibinfo {year} {2015})}\BibitemShut {NoStop}%
\bibitem [{\citenamefont {Parsons}\ \emph {et~al.}(2015)\citenamefont
  {Parsons}, \citenamefont {Huber}, \citenamefont {Mazurenko}, \citenamefont
  {Chiu}, \citenamefont {Setiawan}, \citenamefont {Wooley-Brown}, \citenamefont
  {Blatt},\ and\ \citenamefont {Greiner}}]{Parsons:2015sr}%
  \BibitemOpen
  \bibfield  {author} {\bibinfo {author} {\bibfnamefont {M.~F.}\ \bibnamefont
  {Parsons}}, \bibinfo {author} {\bibfnamefont {F.}~\bibnamefont {Huber}},
  \bibinfo {author} {\bibfnamefont {A.}~\bibnamefont {Mazurenko}}, \bibinfo
  {author} {\bibfnamefont {C.~S.}\ \bibnamefont {Chiu}}, \bibinfo {author}
  {\bibfnamefont {W.}~\bibnamefont {Setiawan}}, \bibinfo {author}
  {\bibfnamefont {K.}~\bibnamefont {Wooley-Brown}}, \bibinfo {author}
  {\bibfnamefont {S.}~\bibnamefont {Blatt}},\ and\ \bibinfo {author}
  {\bibfnamefont {M.}~\bibnamefont {Greiner}},\ }\bibfield  {title} {\bibinfo
  {title} {{Site-Resolved Imaging of Fermionic ${}^6\mathrm{Li}$ in an Optical
  Lattice}},\ }\href {https://doi.org/10.1103/PhysRevLett.114.213002}
  {\bibfield  {journal} {\bibinfo  {journal} {Phys. Rev. Lett.}\ }\textbf
  {\bibinfo {volume} {114}},\ \bibinfo {pages} {213002} (\bibinfo {year}
  {2015})}\BibitemShut {NoStop}%
\bibitem [{\citenamefont {Gross}\ and\ \citenamefont
  {Bakr}(2021)}]{Gross:2021qg}%
  \BibitemOpen
  \bibfield  {author} {\bibinfo {author} {\bibfnamefont {C.}~\bibnamefont
  {Gross}}\ and\ \bibinfo {author} {\bibfnamefont {W.~S.}\ \bibnamefont
  {Bakr}},\ }\bibfield  {title} {\bibinfo {title} {{Quantum gas microscopy for
  single atom and spin detection}},\ }\href
  {https://doi.org/10.1038/s41567-021-01370-5} {\bibfield  {journal} {\bibinfo
  {journal} {Nat. Phys.}\ }\textbf {\bibinfo {volume} {17}},\ \bibinfo {pages}
  {1316} (\bibinfo {year} {2021})}\BibitemShut {NoStop}%
\bibitem [{\citenamefont {Wang}\ \emph {et~al.}(2015)\citenamefont {Wang},
  \citenamefont {Zhang}, \citenamefont {Corcovilos}, \citenamefont {Kumar},\
  and\ \citenamefont {Weiss}}]{Wang:2015ca}%
  \BibitemOpen
  \bibfield  {author} {\bibinfo {author} {\bibfnamefont {Y.}~\bibnamefont
  {Wang}}, \bibinfo {author} {\bibfnamefont {X.}~\bibnamefont {Zhang}},
  \bibinfo {author} {\bibfnamefont {T.~A.}\ \bibnamefont {Corcovilos}},
  \bibinfo {author} {\bibfnamefont {A.}~\bibnamefont {Kumar}},\ and\ \bibinfo
  {author} {\bibfnamefont {D.~S.}\ \bibnamefont {Weiss}},\ }\bibfield  {title}
  {\bibinfo {title} {{Coherent Addressing of Individual Neutral Atoms in a 3D
  Optical Lattice}},\ }\href {https://doi.org/10.1103/PhysRevLett.115.043003}
  {\bibfield  {journal} {\bibinfo  {journal} {Phys. Rev. Lett.}\ }\textbf
  {\bibinfo {volume} {115}},\ \bibinfo {pages} {043003} (\bibinfo {year}
  {2015})}\BibitemShut {NoStop}%
\bibitem [{\citenamefont {Xia}\ \emph {et~al.}(2015)\citenamefont {Xia},
  \citenamefont {Lichtman}, \citenamefont {Maller}, \citenamefont {Carr},
  \citenamefont {Piotrowicz}, \citenamefont {Isenhower},\ and\ \citenamefont
  {Saffman}}]{Xia:2015rb}%
  \BibitemOpen
  \bibfield  {author} {\bibinfo {author} {\bibfnamefont {T.}~\bibnamefont
  {Xia}}, \bibinfo {author} {\bibfnamefont {M.}~\bibnamefont {Lichtman}},
  \bibinfo {author} {\bibfnamefont {K.}~\bibnamefont {Maller}}, \bibinfo
  {author} {\bibfnamefont {A.~W.}\ \bibnamefont {Carr}}, \bibinfo {author}
  {\bibfnamefont {M.~J.}\ \bibnamefont {Piotrowicz}}, \bibinfo {author}
  {\bibfnamefont {L.}~\bibnamefont {Isenhower}},\ and\ \bibinfo {author}
  {\bibfnamefont {M.}~\bibnamefont {Saffman}},\ }\bibfield  {title} {\bibinfo
  {title} {{Randomized Benchmarking of Single-Qubit Gates in a 2D Array of
  Neutral-Atom Qubits}},\ }\href
  {https://doi.org/10.1103/PhysRevLett.114.100503} {\bibfield  {journal}
  {\bibinfo  {journal} {Phys. Rev. Lett.}\ }\textbf {\bibinfo {volume} {114}},\
  \bibinfo {pages} {100503} (\bibinfo {year} {2015})}\BibitemShut {NoStop}%
\bibitem [{\citenamefont {Wang}\ \emph {et~al.}(2016)\citenamefont {Wang},
  \citenamefont {Kumar}, \citenamefont {Wu},\ and\ \citenamefont
  {Weiss}}]{Wang:2016sq}%
  \BibitemOpen
  \bibfield  {author} {\bibinfo {author} {\bibfnamefont {Y.}~\bibnamefont
  {Wang}}, \bibinfo {author} {\bibfnamefont {A.}~\bibnamefont {Kumar}},
  \bibinfo {author} {\bibfnamefont {T.-Y.}\ \bibnamefont {Wu}},\ and\ \bibinfo
  {author} {\bibfnamefont {D.~S.}\ \bibnamefont {Weiss}},\ }\bibfield  {title}
  {\bibinfo {title} {{Single-qubit gates based on targeted phase shifts in a 3D
  neutral atom array}},\ }\href {https://doi.org/10.1126/science.aaf2581}
  {\bibfield  {journal} {\bibinfo  {journal} {Science}\ }\textbf {\bibinfo
  {volume} {352}},\ \bibinfo {pages} {1562} (\bibinfo {year}
  {2016})}\BibitemShut {NoStop}%
\bibitem [{\citenamefont {Moudgalya}\ \emph
  {et~al.}(2020{\natexlab{b}})\citenamefont {Moudgalya}, \citenamefont
  {Regnault},\ and\ \citenamefont {Bernevig}}]{Moudgalya:2020ep}%
  \BibitemOpen
  \bibfield  {author} {\bibinfo {author} {\bibfnamefont {S.}~\bibnamefont
  {Moudgalya}}, \bibinfo {author} {\bibfnamefont {N.}~\bibnamefont
  {Regnault}},\ and\ \bibinfo {author} {\bibfnamefont {B.~A.}\ \bibnamefont
  {Bernevig}},\ }\bibfield  {title} {\bibinfo {title} {{$\eta$-pairing in
  Hubbard models: From spectrum generating algebras to quantum many-body
  scars}},\ }\href {https://doi.org/10.1103/PhysRevB.102.085140} {\bibfield
  {journal} {\bibinfo  {journal} {Phys. Rev. B}\ }\textbf {\bibinfo {volume}
  {102}},\ \bibinfo {pages} {085140} (\bibinfo {year}
  {2020}{\natexlab{b}})}\BibitemShut {NoStop}%
\end{thebibliography}%
\end{document}